\documentclass[a4paper,11pt]{article}

\pdfoutput=1

\usepackage[T1]{fontenc}
\usepackage[utf8]{inputenc}

\usepackage{bm}

\usepackage[english]{babel}

\usepackage{graphicx}
\usepackage{subcaption}
\usepackage{listings}

\usepackage[version=4]{mhchem}  
\usepackage{amsthm, amssymb}
\usepackage{siunitx}

\usepackage[margin=25mm]{geometry}
\usepackage[skins, breakable]{tcolorbox}

\usepackage[square, numbers, sort&compress]{natbib}

\usepackage{catchfile}  
\usepackage{doi}  
\hypersetup{
	linktocpage=true,
	bookmarksopen=true,
	bookmarksnumbered=true,
}

\usepackage{setspace}
\usepackage{lineno}
\usepackage{todonotes}

\begin{document}


\newcommand{\citex}[2]{\cite[\expandafter{#1}]{#2}}  
\newcommand{\dlmf}[3][]{\href{http://dlmf.nist.gov/#2#1#3}{#2#3}}
\newcommand{\wikipedia}[2][en]{\href{http://#1.wikipedia.org/wiki/#2}{#2}}
\newcommand{\elektrodo}{https://gitlab.com/cfgy/elektrodo/tree/publ-2018jan}

\newcommand{\fourier}{\operatorname{\mathcal{F}}_{x}}
\newcommand{\laplace}{\operatorname{\mathcal{L}}_{t}}
\newcommand{\ud}[1]{\,\mathrm{d}#1}

\newcommand{\e}{\operatorname{e}}

\newcommand{\flechadoble}{\ce{<=>}}
\newcommand{\deriv}[2]{\frac{\ud #1}{\ud #2}}
\newcommand{\parderiv}[2]{\frac{\partial #1}{\partial #2}}
\newcommand{\simu}{_{\mathrm{sim}}}
\newcommand{\unit}{^{\text{unit}}}
\newcommand{\whole}{^{\text{whole}}}

\newtheorem{observacion}{Remark}[section]
\newtheorem{definicion}{Definition}[section]
\newtheorem{lema}{Lemma}[section]
\newtheorem{teorema}{Theorem}[section]
\newtheorem{corolario}{Corollary}[section]

\tcolorboxenvironment{observacion}{enhanced jigsaw, breakable, before skip=1em, after skip=1em, colframe=gray}
\tcolorboxenvironment{definicion}{enhanced jigsaw, breakable, before skip=1em, after skip=1em, colframe=gray}
\tcolorboxenvironment{lema}{enhanced jigsaw, breakable, before skip=1em, after skip=1em, colframe=gray}
\tcolorboxenvironment{teorema}{enhanced jigsaw, breakable, before skip=1em, after skip=1em, colframe=gray}
\tcolorboxenvironment{corolario}{enhanced jigsaw, breakable, before skip=1em, after skip=1em, colframe=gray}
\tcolorboxenvironment{proof}{blanker, breakable, left=5mm, before skip=1em, after skip=1em, borderline west={1mm}{0pt}{gray}, parbox=false}
\renewcommand{\qedsymbol}{\textit{QED}.}


\newcommand{\givennames}[1]{#1}
\newcommand{\familynames}[1]{\textsc{#1}}

\title{
	Properties of coplanar periodic electrodes in confined spaces:
	Case of two-dimensional diffusion
}

\author{
	\givennames{Cristian F.} \familynames{Guajardo Yévenes} \\
	Biological Engineering Program \\
	and Pilot Plant, Development and Training Institute \\
	King Mongkut's University of Technology Thonburi, Thailand \\
	\texttt{cristian.gua@kmutt.ac.th}
	\and
	\givennames{Werasak} \familynames{Surareungchai} \\
	School of Bioresources and Technology \\
	and Nanoscience \& Nanotechnology Graduate Program \\
	King Mongkut's University of Technology Thonburi, Thailand \\
	\texttt{werasak.sur@kmutt.ac.th}
}

\date{February 10, 2018}


\newcommand{\sep}{, }
\newcommand{\resumen}{Periodic configurations of electrodes, in particular of microelectrodes,
have been of interest since the advent of microfabrication.
In this report, theory which is common to any periodic cell
(or any cell that can be extended periodically)
with finite height and two-dimensional symmentry was derived.
The diffusion equation in this cell was solved
and the concentration profile was obtained
in terms of its Fourier coefficients
and as a function of an arbitrary current density.
From this base result, a set of properties were derived
which are fairly general, since they don't assume restrictions
such as reversible electrode reactions
(Nernst equation valid when current circulates).
These properties involve:
horizontal averages and (weighted) sum of concentrations,
both with a close connection to the net current
and accumulation of species in the cell.
The derived properties allow:
to explain qualitative aspects of collection efficiency and limiting currents,
to predict the concentration on counter electrodes and non-linearities
caused by depletion of species at extremely polarized electrodes,
and to estimate the time required by the current to reach steady state
in potential controlled experiments.
The theoretical results are illustrated analytically
and numerically for the concrete case of interdigitated array of electrodes.
}
\newcommand{\pclaves}{Periodic cell\sep
confined cell\sep
diffusion equation\sep
average concentration\sep
collection efficiency\sep
limiting current.
}
\hypersetup{
	pdftitle={Properties of coplanar periodic electrodes in confined spaces: Case of two-dimensional diffusion},
	pdfauthor={Cristian F. GUAJARDO YÉVENES and Werasak SURAREUNGCHAI},
	pdfsubject={Periodic configurations of electrodes, in particular of microelectrodes,
have been of interest since the advent of microfabrication.
In this report, theory which is common to any periodic cell
(or any cell that can be extended periodically)
with finite height and two-dimensional symmentry was derived.
The diffusion equation in this cell was solved
and the concentration profile was obtained
in terms of its Fourier coefficients
and as a function of an arbitrary current density.
From this base result, a set of properties were derived
which are fairly general, since they don't assume restrictions
such as reversible electrode reactions
(Nernst equation valid when current circulates).
These properties involve:
horizontal averages and (weighted) sum of concentrations,
both with a close connection to the net current
and accumulation of species in the cell.
The derived properties allow:
to explain qualitative aspects of collection efficiency and limiting currents,
to predict the concentration on counter electrodes and non-linearities
caused by depletion of species at extremely polarized electrodes,
and to estimate the time required by the current to reach steady state
in potential controlled experiments.
The theoretical results are illustrated analytically
and numerically for the concrete case of interdigitated array of electrodes.
},
	pdfkeywords={Periodic cell\sep
confined cell\sep
diffusion equation\sep
average concentration\sep
collection efficiency\sep
limiting current.
}
}


\numberwithin{equation}{section}

\maketitle

\section*{Abstract}


\vspace{\baselineskip} \noindent
\emph{Keywords:}

\tableofcontents


\section{Introduction}

Microelectrodes have been used since early 1980 \cite{Dayton:1980:}
due to their many advantageous properties, such as reduced ohmic drops,
faster time constants, better signal-to-noise ratios
and steady-state signals \cite{Forster:2007:,Szunerits:2007:}.
These electrodes have also been arranged in a periodic fashion (arrays),
in order to produce higher currents, while still maintaining
the basic microelectrode properties \cite{Szunerits:2007:}.
From these periodic configurations, \emph{microband array electrodes} (MBAE, only anodes or cathodes)
and \emph{interdigitated array of electrodes} (IDAE, alternating anodes and cathodes)
are common examples found in the literature.

Theoretical results for periodic configurations were first obtained
considering unrestricted (semi-infinite) geometries.
Analytical results to predict steady-state currents and voltammograms were found
in case of MBAE \cite{Morf:1996:sep,Morf:2006:may} and
in case of IDAE \cite{Aoki:1988:dec,Aoki:1990:apr,Morf:2006:may}.
Numerical results through simulations have been also obtained
to estimate the time dependence of the current and voltammograms
in case of MBAE \cite{Bard:1986:sep,Streeter:2007:aug,Pebay:2013:dec} and
in case of IDAE \cite{Aoki:1989:jul,Jin:1996:aug:b,Yang:2007:oct}.
Besides these mature results, there are also novel semi-analytical results
predicting the chronoamperometry at microband electrodes \cite{Bieniasz:2015:oct}.

Currently, with the advent of microfluidic technology and flexible materials,
confined (finite) electrochemical cells have gained importance,
since electrochemical cells are placed inside shallow channels \cite{Han:2014:}
or meant to be used in narrow cavities of the body \cite{Kanno:2014:}.

In the literature, the behavior of such electrodes
in restricted or finite spaces has been predicted mostly through simulations,
which allows interpretation of electrochemical phenomena
in case of IDAE \cite{Strutwolf:2005:feb,Goluch:2009:may,Han:2014:,Kanno:2014:}
and in case of MBAE \cite{Bellagha-Chenchah:2016:jun}.
Analytical results to predict the behavior in confined spaces are few
\cite{GuajardoYevenes:2013:sep},
and commonly the results for semi-infinite counterparts are used instead \cite{Shim:2013:},
which are valid only when the cell is tall enough \cite{GuajardoYevenes:2013:sep}.

In this report, analytical properties which are common to any periodic cell
(or any cell that can be extended periodically) with finite height 
and two-dimensional symmetry are derived.
The base analytical result consists of
the concentration profile in stagnant solution,
expresed in terms of its Fourier coeficients,
and considers an arbitrary current density flowing in the cell.
From this result, analytical properties for average concentrations
and weighted sum of concentrations are derived.

These results are of importance since they can explain qualitative aspects
of collection efficiency and limiting currents.
Also they allow to determine the concentration of electrochemical species
on the counter electrode (commonly unknown \emph{a priori}),
which is particularly useful for defining boundary conditions
of simulations that include such electrode.
Explanation of non-linear effects on the concentration
caused by depletion of species at electrodes that are extremely polarized,
is also possible with the results.
Finally estimations of the time required by the current to reach steady state can be obtained.
These properties are illustrated analytically and numerically by simulations
for the particular case of interdigitated array of electrodes.


\section{Theory}
\label{theory}

\subsection{Definition of the periodic cell}
\label{coplanar:def:problema}

Consider an electrochemical cell with a coplanar configuration of electrodes
located at the bottom plane $z=0$,
and a roof (insulator layer) located at the top plane $z=H$.
The configuration of electrodes is periodic (with period $p_{x}$) along the $x$-axis,
and it is symmetric along the $y$-axis,
such that the concentration profiles don't depend on the variable $y$.

Inside the electrochemical cell there is an oxidated species $O$ and a reduced species $R$,
which react at the surface of the electrodes according to the reaction
\begin{equation}
	\label{coplanar:eqn:reaccion}
	O + n_{e}\, \mathrm{e}^{-} \flechadoble R
\end{equation}
where $n_{e}$ corresponds to the number of exchanged electrons.
Here it is assumed that the transport of the species $\sigma \in \{O,R\}$ is solely due to diffusion.

Under the stated conditions,
the concentration $c_{\sigma}(x,z,t)$ of the electrochemical species $\sigma$
can be modeled by the two-di\-men\-sio\-nal diffusion equation%
\footnote{
	whenever $\pm$ or $\mp$ are found,
	the upper and lower signs corresponds to 
	$\sigma=O$ and $\sigma=R$ respectively.
}
\begin{subequations}
	\label{coplanar:eqn:difusion:idae}
	\begin{align}
		\frac{1}{D_{\sigma}} \parderiv{c_{\sigma}}{t}(x,z,t)
		&= \parderiv{^2 c_{\sigma}}{x^2}(x,z,t)
		+ \parderiv{^2 c_{\sigma}}{z^2}(x,z,t)
		\\
		c_{\sigma}(x,z,0^{-}) &= c_{\sigma,i}(x,z)
		\\
		\mp D_{\sigma} \parderiv{c_{\sigma}}{z}(x,H,t) &= 0
		\\
		\mp D_{\sigma} \parderiv{c_{\sigma}}{z}(x,0,t) &= \frac{j(x,t)}{F n_{e}}
		\\
		c_{\sigma}(x, z, t) &= c_{\sigma}(x + p_{x}, z, t)
	\end{align}
\end{subequations}
where $c_{\sigma,i}(x,z)$ is the initial concentration profile,
which is assumed to come from a previous steady state,
$D_{\sigma}$ is the diffusion coefficient of the species $\sigma$,
$F$ is the Faraday's constant,
$j(x,t)$ is an \textbf{arbitrary} current density flowing at the bottom boundary $z=0$,
and $p_{x}$ is the period along the $x$-axis.

\subsection{Properties of the initial concentration in steady state}
\label{coplanar:properties:initial}

With the assumptions in \S\ref{coplanar:def:problema},
it is possible to derive some useful \emph{conservation} properties,
which araise as direct consequences of the theorem stated below

\begin{teorema}
	\label{coplanar:teo:ci}
	Consider the periodic cell described in \S\ref{coplanar:def:problema}.
	If the initial concentration $c_{\sigma,i}(x,z)$ of species $\sigma \in \{O, R\}$
	comes from a previous steady state,
	then the initial current density $j_{i}(x)$ satisfies Kirchhoff's current law
	within one period of the cell
	\begin{equation}
		\label{coplanar:eqn:kirchhoff:inicial}
		p_{x} \fourier j_{i}(0)
		= \int_{-p_{x}/2}^{+p_{x}/2} j_{i}(x) \ud{x} = 0
	\end{equation}
	and the Fourier coefficients of the initial concentration profile satisfy
	\begin{equation}
		\label{coplanar:eqn:ci}
		\fourier c_{\sigma,i}(n_{x},z) = \left\{
		\begin{array}{ll}
		\bar{c}_{\sigma,i}, & n_{x} = 0
		\\
		\displaystyle
		\pm G\!\left(H-z,\, n_{x}^{2} \frac{4\pi^{2}}{p_{x}^{2}}\right)
		\frac{\fourier j_{i}(n_{x})}{F n_{e} D_{\sigma}}, & n_{x} \neq 0
		\end{array}
		\right.
	\end{equation}
	where $\fourier j_{i}(n_{x})$ are the Fourier coefficients of the initial current density,
	$p_{x}$ is one period of the cell, $\bar{c}_{\sigma,i}$ is a real constant
	and $G(z,s)$ is given by
	\begin{equation}
		\label{coplanar:eqn:G}
		G(z,s) = \frac{\cosh(\sqrt{s}\, z)}{\sqrt{s} \sinh(\sqrt{s}\, H)}
	\end{equation}
\end{teorema}
See Supplementary Information \S\ref{transforms} for the definition used for the Fourier coefficients.

\begin{proof}
	Since the initial concentration comes from a previous steady state,
	it satifies the diffusion equation in Eqs. \eqref{coplanar:eqn:difusion:idae}
	with $\partial c_{\sigma}/ \partial t = 0$
	\begin{subequations}
		\label{coplanar:eqn:difusion:estacionaria}
		\begin{align}
			\parderiv{^2 c_{\sigma,i}}{x^2}(x,z) +
			\parderiv{^2 c_{\sigma,i}}{z^2}(x,z) &= 0
			\\
			\mp D_{\sigma} \parderiv{c_{\sigma,i}}{z}(x,H) &= 0
			\\
			\mp D_{\sigma} \parderiv{c_{\sigma,i}}{z}(x,0) &= \frac{j_{i}(x)}{F n_{e}}
			\\
			c_{\sigma,i}(x, z) &= c_{\sigma,i}(x + p_{x}, z)
		\end{align}
	\end{subequations}
	Taking the Fourier coefficients $\fourier c_{\sigma,i}(n_{x},z)$
	from the diffusion equation, one obtains
	\begin{subequations}
		\begin{align}
			- n_{x}^{2} \frac{4\pi^{2}}{p_{x}^{2}} \fourier c_{\sigma,i}(n_{x},z)
			+ \parderiv{^2 \fourier c_{\sigma,i}}{z^2}(n_{x},z) &= 0 \\
			\mp D_{\sigma} \parderiv{\fourier c_{\sigma,i}}{z}(n_{x},H) &= 0 \\
			\mp D_{\sigma} \parderiv{\fourier c_{\sigma,i}}{z}(n_{x},0)
			&= \frac{\fourier j_{i}(n_{x})}{F n_{e}}
		\end{align}
	\end{subequations}
	which corresponds to a linear ordinary differential equation (ODE),
	thus it can be solved using well known techniques.

	In case $n_{x} \neq 0$, solving the ODE in terms of the Fourier coefficients leads to
	\begin{subequations}
		\begin{equation}
			\fourier c_{\sigma,i}(n_{x}, z) =
			a(n_{x}) \cosh\!\left(n_{x} \frac{2\pi}{p_{x}} (H-z)\right) +
			b(n_{x}) \sinh\!\left(n_{x} \frac{2\pi}{p_{x}} (H-z)\right)
		\end{equation}
		Later, by applying the boundary conditions, the desired result is obtained
		\begin{equation}
			\fourier c_{\sigma,i}(n_{x}, z) =
			\pm G\!\left(H-z,\, n_{x}^{2} \frac{4\pi^{2}}{p_{x}^{2}}\right)
			\frac{\fourier j_{i}(n_{x})}{F n_{e} D_{\sigma}}
		\end{equation}
	\end{subequations}
	where $G(z,s)$ is defined in Eq. (\ref{coplanar:eqn:G}).

	In case $n_{x} = 0$, the solution of the ODE in terms of $\fourier c_{\sigma,i}(0, z)$
	leads to a real constant independent of $z$, name it $\bar{c}_{\sigma,i}$,
	and the Fourier coefficient of the current density equals zero due to Fick's law
	\begin{equation}
		\fourier c_{\sigma,i}(0,z) = \bar{c}_{\sigma,i},
		\quad
		\fourier j_{i}(0) = 0
	\end{equation}
\end{proof}

The first \emph{conservation} property that can be obtained from the previous theorem
holds for the horizontal average of $c_{\sigma,i}(x,z)$ at any $z$ in the cell

\begin{corolario}
	\label{coplanar:cor:bar_ci}
	Assume a two-dimensional periodic cell as in \S\ref{coplanar:def:problema},
	where the initial concentration $c_{\sigma,i}(x,z)$ of species $\sigma \in \{O, R\}$
	comes from a previous steady state.
	The average of the initial concentration, along any horizontal line,
	is independent of $z$ and equals $\bar{c}_{\sigma,i}$
	\begin{equation}
		\label{coplanar:eqn:bar_ci}
		\frac{1}{p_{x}} \int_{-p_{x}/2}^{+p_{x}/2} c_{\sigma,i}(x,z) \ud{x} =
		\fourier c_{\sigma,i}(0,z) = \bar{c}_{\sigma,i}
	\end{equation}
	where $p_{x}$ is one period of the cell.
\end{corolario}

The second \emph{conservation} property holds for the weighted sum of concentrations at any point in the cell,
which translates into the conservation of the total concentration at any point in the cell
when the diffusion coefficients of both electrochemical species are equal

\begin{corolario}
	\label{coplanar:cor:ci:total}
	Assume a two-dimensional periodic cell with period $p_{x}$ as in \S\ref{coplanar:def:problema},
	where the initial concentration $c_{\sigma,i}(x,z)$ comes from a previous steady state. The following weighted sum of the initial concentrations is independent of $(x,z)$ and equals
	\begin{equation}
		\label{coplanar:eqn:ci_total}
		D_{O} c_{O,i}(x,z) + D_{R} c_{R,i}(x,z) =
		D_{O} \bar{c}_{O,i} + D_{R} \bar{c}_{R,i}
	\end{equation}
\end{corolario}

\begin{proof}
	Take the weighted sum of both Fourier coefficients
	\begin{equation}
		D_{O} \fourier c_{O,i}(n_{x},z) + D_{R} \fourier c_{R,i}(n_{x},z) =
		\left\{
			\begin{array}{ll}
				D_{O} \bar{c}_{O,i} + D_{R} \bar{c}_{R,i}, & n_{x} = 0 \\
				0, & n_{x} \neq 0
			\end{array}
		\right.
	\end{equation}
	and later, take its Fourier series.
\end{proof}

\subsection{Properties of the concentration in transient state}

By using the Laplace transform and the Fourier coefficients
on the change in concentration $\Delta c_{\sigma}(x,z,t) = c_{\sigma}(x,z,t) - c_{\sigma,i}(x,z)$,
one can derive similar properties as in the previous section,
but now for the transient state.

\begin{teorema}
	\label{coplanar:teo:Dc}
	Consider the periodic cell described in \S\ref{coplanar:def:problema}.
	If the initial concentration $c_{\sigma,i}(x,z)$ of species $\sigma \in \{O, R\}$
	comes from a previous steady state,
	then the Laplace transform of the Fourier coefficients
	of $\Delta c_{\sigma}(x,z,t) = c_{\sigma}(x,z,t) - c_{\sigma,i}(x,z)$ is given by
	\begin{equation}
		\label{coplanar:eqn:Dc}
		\laplace \fourier \Delta c_{\sigma}(n_{x},z,s) =
		\pm G\!\left(H-z,\, \frac{s}{D_{\sigma}} + n_{x}^{2} \frac{4\pi^{2}}{p_{x}^{2}}\right)
		\frac{\laplace \fourier \Delta j(n_{x},s)}{F n_{e} D_{\sigma}}
	\end{equation}
	where $p_{x}$ is one period of the cell,
	$G(z,s)$ is defined in Eq. (\ref{coplanar:eqn:G}),
	and $\laplace\fourier \Delta j(n_{x},s)$ is the Laplace transform
	of the Fourier coefficients of $\Delta j(x,t) = j(x,t) - j_{i}(x)$.
\end{teorema}
See Supplementary Information \S\ref{transforms} for the definitions of the Fourier coefficients and the Laplace transform used in the previous theorem.

\begin{proof}
	First, substract Eqs. (\ref{coplanar:eqn:difusion:idae})
	and (\ref{coplanar:eqn:difusion:estacionaria})
	to obtain the following partial differential equation
	\begin{subequations}
		\begin{align}
			\frac{1}{D_{\sigma}} \parderiv{\Delta c_{\sigma}}{t}(x,z,t)
			&= \parderiv{^2 \Delta c_{\sigma}}{x^2}(x,z,t)
			+ \parderiv{^2 \Delta c_{\sigma}}{z^2}(x,z,t)
			\\
			\Delta c_{\sigma}(x,z,0^{-}) &= 0
			\\
			\mp D_{\sigma} \parderiv{\Delta c_{\sigma}}{z}(x,H,t) &= 0
			\\
			\mp D_{\sigma} \parderiv{\Delta c_{\sigma}}{z}(x,0,t)
			&= \frac{\Delta j(x,t)}{F n_{e}}
			\\
			\Delta c_{\sigma}(x, z, t) &= \Delta c_{\sigma}(x + p_{x}, z, t)
		\end{align}
	\end{subequations}
	which depends on the changes of concentration and current density
	with respect to the initial condition.

	By taking the Fourier coefficients in $x$ and the Laplace transform in $t$,
	one can convert this problem into an ordinary differential equation
	\begin{subequations}
		\begin{align}
			\left(
				\frac{s}{D_{\sigma}} + n_{x}^{2} \frac{4\pi^{2}}{p_{x}^{2}}
			\right)
			\laplace \fourier \Delta c_{\sigma}(n_{x},z,s)
			&= \parderiv{^2 \laplace \fourier \Delta c_{\sigma}}{z^2}(n_{x},z,s)
			\\
			\mp D_{\sigma} \parderiv{\laplace \fourier \Delta c_{\sigma}}{z}(n_{x},H,s)
			&= 0
			\\
			\mp D_{\sigma} \parderiv{\laplace \fourier \Delta c_{\sigma}}{z}(n_{x},0,s)
			&= \frac{\laplace \fourier \Delta j(n_{x},s)}{F n_{e}}
		\end{align}
	\end{subequations}
	of which its solution
	\begin{multline}
		\laplace \fourier \Delta c_{\sigma}(n_{x},z,s) =
		A(n_{x},s) \cosh\!\left(
			\sqrt{\frac{s}{D_{\sigma}} + n_{x}^{2} \frac{4\pi^{2}}{p_{x}^{2}}}
			(H - z)
		\right)
		\\
		+ B(n_{x},s) \sinh\!\left( 
			\sqrt{\frac{s}{D_{\sigma}} + n_{x}^{2} \frac{4\pi^{2}}{p_{x}^{2}}}
			(H - z)
		\right)
	\end{multline}
	after applying the boundary conditions, 
	is given by Eq. (\ref{coplanar:eqn:Dc}),
	where $G(z,s)$ is defined in Eq. (\ref{coplanar:eqn:G}).
\end{proof}

Before obtaining the properties for the concentration in transient state,
it is useful to obtain the time-domain counterparts of the frequency-domain function $G(z,s)$ in Eq. \eqref{coplanar:eqn:G}.
\begin{lema}
	\label{coplanar:lem:g-h}
	Consider the transfer function $G(z,s)$
	which is defined in Eq. (\ref{coplanar:eqn:G}).
	The inverse Laplace transform $g(z,t) = \laplace^{-1} G(z,s)$
	is given by
	\begin{equation}
		\label{coplanar:eqn:g}
		g(z,t)
		= \frac{1}{H}
		\left[
			1 + 2 \sum_{k=1}^{+\infty} (-1)^{k}
			\cos\!\left( k \frac{\pi}{H} z \right)
			\exp\!\left( -k^{2} \frac{\pi^{2}}{H^{2}} t \right)
		\right]
	\end{equation}
	where the argument of its exponential factors correspond to the poles of $G(z,s)$.
	Note that $g(z,t) = H^{-1} \theta_{4}(z\pi/2H|\bm{i}\pi t/H^{2})$ 
	\citex{Eq. (\dlmf[E]{20.10.}{5})}{dlmf} is related to the \emph{elliptic theta function}
	$\theta_{4}(z|\tau) = \theta_{4}(z,q)$ \citex{Eq. (\dlmf[E]{20.2.}{4})}{dlmf}
	where $q = \exp(\bm{i}\pi\tau)$ \cite[\S\dlmf{20.}{1}]{dlmf}.

	And the inverse Laplace transform $h(z,t) = \laplace^{-1} G(z,s)^{-1}$ is given by
	\begin{equation}
		\label{coplanar:eqn:h}
		h(z,t) =
		\frac{2}{z} \sum_{\ell=1}^{+\infty} (2\ell - 1)^{2} \frac{\pi^{2}}{4 z^{2}} (-1)^{\ell}
		\sin\!\left( (2\ell - 1) \frac{\pi H}{2z} \right)
		\exp\!\left( -(2\ell - 1)^{2} \frac{\pi^{2}}{4 z^{2}} t \right)
	\end{equation}
	where the argument of its exponential factors corresponds to the zeros of $G(z,s)$.
	Note that $h(z,t) = z^{-1} \dot{\theta}_{1}(H\pi/2z|\bm{i}\pi t/z^{2})$
	\citex{Eq. (\dlmf[E]{20.10.}{4})}{dlmf} 
	is related to time derivative of the \emph{elliptic theta function}
	$\theta_{1}(z|\tau) = \theta_{1}(z,q)$ \citex{Eq. (\dlmf[E]{20.2.}{1})}{dlmf}
	where $q = \exp(\bm{i}\pi\tau)$ \cite[\S\dlmf{20.}{1}]{dlmf}.
\end{lema}

\begin{proof}
	From \citex{Eq. (\dlmf[E]{20.10.}{5})}{dlmf}
	\begin{equation}
		\laplace \theta_{4}\!\left( \frac{z\pi}{2H} \left| \frac{\bm{i}\pi t}{H^{2}} \right. \right)
		= H G(z,s)
	\end{equation}
	Finally, we let $g(z,t) = H^{-1} \theta_{4}(z\pi/2H|\bm{i}\pi t/H^{2})$.
	
	From \citex{Eq. (\dlmf[E]{20.10.}{4})}{dlmf} and
	the property of the Laplace transform of the time derivative
	\begin{equation}
		\laplace \deriv{}{t} \theta_{1}\!\left( \frac{H\pi}{2z} \left| \frac{\bm{i}\pi t}{z^{2}} \right. \right)
		= s\, \laplace \theta_{1}\!\left( \frac{H\pi}{2z} \left| \frac{\bm{i}\pi t}{z^{2}} \right. \right)
		= z\, G(z,s)^{-1}
	\end{equation}
	Finally, we let $h(z,t) = z^{-1} \dot{\theta}_{1}(H\pi/2z|\bm{i}\pi t/z^{2})$.
\end{proof}

The first transient property that can be obtained from Theorem \ref{coplanar:teo:Dc}
holds for the horizontal average of $c_{\sigma}(x,z,t)$.
Note that unlike Corollary \ref{coplanar:cor:bar_ci},
the Corollary below shows that the horizontal average is not uniform along $z$
and also changes with time.

\begin{corolario}
	\label{coplanar:cor:delta_bar_c}
	Consider the periodic cell described in \S\ref{coplanar:def:problema}
	and assume that the initial concentration $c_{\sigma,i}(x,z)$ comes from a previous steady state.

	The average of the concentration, along any horizontal line, equals
	\begin{subequations}
		\label{coplanar:eqn:bar_c}
		\begin{gather}
			\frac{1}{p_{x}}
			\int_{-p_{x}/2}^{+p_{x}/2} c_{\sigma}(x,z,t) \ud{x}
			= \bar{c}_{\sigma,i} \pm \Delta \bar{c}_{\sigma}(z,t)
			\\
			\intertext{%
				where the change in average concentration $\Delta\bar{c}_{\sigma}(z,t)$
				depends on $z$, $t$, the electrochemical species $\sigma$,
				and the net current in a period of the cell
			}
			\label{coplanar:eqn:delta_bar_c}
			\Delta \bar{c}_{\sigma}(z,t)
			= g(H-z, D_{\sigma}t) * \frac{1}{p_{x}}
			\int_{-p_{x}/2}^{+p_{x}/2} \frac{j(x,t)}{F n_{e}} \ud{x}
		\end{gather}
	\end{subequations}
	and where $g(z,t) = \laplace^{-1} G(z,s)$ is given in Eq. (\ref{coplanar:eqn:g}).

	Conversely,
	the average current density (net current) in one period of the cell
	is dependent on $t$,
	and on the change in average concentration $\Delta\bar{c}_{\sigma}(0,t)$
	at the bottom of cell (where the electrodes are located)
	\begin{equation}
		\label{coplanar:eqn:i_net}
		\frac{1}{p_{x}} \int_{-p_{x}/2}^{+p_{x}/2} j(x,t) \ud{x}
		= F n_{e} D_{\sigma}^{2}\, h(H, D_{\sigma} t) * \Delta\bar{c}_{\sigma} (0,t)
	\end{equation}
	where $h(z,t) = \laplace^{-1} G(s,z)^{-1}$ is given by Eq. (\ref{coplanar:eqn:h}).

	In both cases, $*$ is the time convolution, $p_{x}$ is one period of the cell, and
	$\bar{c}_{\sigma,i}$ is the horizontal average of the initial concentration,
	see Eq. (\ref{coplanar:eqn:bar_ci}).
\end{corolario}

\begin{proof}
	Take the expression for $\laplace\fourier \Delta c_{\sigma}(0,z,s)$ from Eq. (\ref{coplanar:eqn:Dc}).
	\begin{equation}
		\frac{1}{p_{x}}
		\int_{-p_{x}/2}^{+p_{x}/2} \laplace \Delta c_{\sigma}(x,z,s) \ud{x} =
		\pm G\!\left( H-z, \frac{s}{D_{\sigma}} \right)
		\frac{1}{p_{x}}
		\int_{-p_{x}/2}^{+p_{x}/2} \frac{\laplace\Delta j(x,s)}{F n_{e} D_{\sigma}} \ud{x}
	\end{equation}
	Let $g(z,t) = \laplace^{-1} G(z,s)$,
	then the previous equation can be written in time domain
	by applying the inverse Laplace transform
	together with the \emph{time scaling} property
	\begin{equation}
		\frac{1}{p_{x}} \int_{-p_{x}/2}^{+p_{x}/2} \Delta c_{\sigma}(x,z,t) \ud{x} =
		 \pm D_{\sigma} g(H-z, D_{\sigma}t) * \frac{1}{p_{x}}
		\int_{-p_{x}/2}^{+p_{x}/2} \frac{\Delta j(x,t)}{F n_{e} D_{\sigma}} \ud{x}
	\end{equation}
	By adding Eq. (\ref{coplanar:eqn:bar_ci}) to the previous equation,
	and later, by applying Eq. (\ref{coplanar:eqn:kirchhoff:inicial}),
	leads to Eqs. (\ref{coplanar:eqn:bar_c}).

	By taking Eq. (\ref{coplanar:eqn:delta_bar_c}) in Laplace domain
	and later by isolating the average current density (net current) one obtains
	\begin{equation}
		\frac{1}{p_{x}} \int_{-p_{x}/2}^{+p_{x}/2} \frac{\laplace j(x,s)}{D_{\sigma} F n_{e}} \ud{x}
		= G\!\left( H - z, \frac{s}{D_{\sigma}} \right)^{-1} \laplace \Delta\bar{c}_{\sigma}(z,s)
	\end{equation}
	Let $h(z,t) = \laplace^{-1} G(z,s)^{-1}$,
	then the previous equation can be written in time domain
	by applying the inverse Laplace transform
	together with the \emph{time scaling} property
	\begin{equation}
		\frac{1}{p_{x}} \int_{-p_{x}/2}^{+p_{x}/2} \frac{j(x,t)}{D_{\sigma} F n_{e}} \ud{x}
		= D_{\sigma}\, h(H - z, D_{\sigma} t) * \Delta\bar{c}_{\sigma} (z,t)
	\end{equation}
	Since the average current density is independent of $z$,
	it suffices to take $z=0$, leading to Eq. (\ref{coplanar:eqn:i_net}).
\end{proof}

The second transient property is a \emph{conservation} property,
and holds for the total concentration at any point in the cell,
and any time $t \geq 0$.

\begin{corolario}
	\label{coplanar:cor:ct:total}
	Consider the periodic cell with period $p_{x}$,
	described in \S\ref{coplanar:def:problema},
	and assume that the initial concentration $c_{\sigma,i}(x,z)$ comes from a previous steady state.
	If the diffusion coefficients of both species are equal $D_{O} = D_{R}$,
	then the sum of the concentrations at any point in the cell is independent of $(x,z,t)$ and equals
	\begin{equation}
		\label{coplanar:eqn:ct:total}
		c_{O}(x,z,t) + c_{R}(x,z,t) = \bar{c}_{O,i} + \bar{c}_{R,i}
	\end{equation}
	where $\bar{c}_{\sigma,i}$ with $\sigma \in \{O,R\}$ is given in Eq. (\ref{coplanar:eqn:bar_ci}).
\end{corolario}

\begin{proof}
	If $D_{O} = D_{R}$, then the sum of Eq. (\ref{coplanar:eqn:Dc}) for both electrochemical species is
	\begin{equation}
		\laplace\fourier \Delta c_{O}(n_{x},z,s)
		+ \laplace\fourier \Delta c_{R}(n_{x},z,s) = 0
	\end{equation}
	Taking the inverse Laplace transform and later the Fourier series, one obtains
	\begin{equation}
		\Delta c_{O}(x,z,t) + \Delta c_{R}(x,z,t) = 0
	\end{equation}
	Finally, by adding Eq. (\ref{coplanar:eqn:ci_total}), Eq. (\ref{coplanar:eqn:ct:total}) is obtained.
\end{proof}

\subsection{Properties of the final concentration in steady state}

\emph{Conservation} properties similar to those in \S\ref{coplanar:properties:initial}
also hold for the final concentration in steady state,
which araise as direct consequences of the theorem stated below.

\begin{teorema}
	\label{coplanar:teo:cf}
	Consider the periodic cell described in \S\ref{coplanar:def:problema}.
	If the initial concentration $c_{\sigma,i}(x,z)$ of species $\sigma \in \{O, R\}$
	comes from a previous steady state and the following integral converges
	\begin{equation}
		\label{coplanar:eqn:i_neta_acumulada}
		\int_{0^{-}}^{+\infty} \int_{-p_{x}/2}^{+p_{x}/2} j(x,t) \ud{x} \ud{t}
	\end{equation}
	then the final current density $j_{f}(x) = \lim_{t \to +\infty} j(x,t)$
	satisfies Kichhoff's current law in one period of the cell
	\begin{equation}
		\label{coplanar:eqn:kirchhoff:final}
		\int_{-p_{x}/2}^{+p_{x}/2} j_{f}(x) \ud{x} = 0
	\end{equation}
	and the Fourier coefficients $\fourier c_{\sigma,f}(n_{x},z)$ of the final concentration 
	$c_{\sigma,f}(x,z) = \lim_{t \to +\infty} c_{\sigma}(x,z,t)$ are given by
	\begin{subequations}
		\begin{align}
			\fourier c_{\sigma,f}(n_{x},z) &=
			\left\{
				\begin{array}{ll}
					\displaystyle
					\bar{c}_{\sigma,f} =
					\bar{c}_{\sigma,i} \pm \Delta \bar{c}_{f}, & n_{x} = 0
					\\[1em]
					\displaystyle
					\pm G\!\left(H-z,\, n_{x}^{2} \frac{4\pi^{2}}{p_{x}^{2}}\right)
					\frac{\fourier j_{f}(n_{x})}{F n_{e} D_{\sigma}}, & n_{x} \neq 0
				\end{array}
			\right.
			\\[0.5em]
			\label{coplanar:eqn:delta_bar_cf}
			\Delta \bar{c}_{f} &= \frac{1}{H}
			\int_{0^{-}}^{+\infty} \frac{1}{p_{x}} \int_{-p_{x}/2}^{+p_{x}/2}
			\frac{j(x,t)}{F n_{e}} \ud{x} \ud{t}
		\end{align}
	\end{subequations}
	where $p_{x}$ is one period of the cell,
	$\bar{c}_{\sigma,i}$ is the horizontal average of the initial concentration
	defined in Eq. (\ref{coplanar:eqn:bar_ci}),
	and $G(z,s)$ is defined in Eq. (\ref{coplanar:eqn:G}).
\end{teorema}

Note from the theorem above that Eq. \eqref{coplanar:eqn:kirchhoff:final}
(that is, Kirchoff's current law be satisfied in steady state, or equivalently,
100\% collection efficiency in the final steady state)
is a \textbf{necessary condition}
for the convergence of the concentration profile in the final steady state.

\begin{proof}
	The final steady state can be obtained
	if one applies the \emph{final value theorem} of the Laplace transform
	to $\laplace \fourier \Delta c_{\sigma}(n_{x},z,s)$ in Eq. (\ref{coplanar:eqn:Dc})
	\begin{equation}
		\fourier \Delta c_{\sigma,f}(n_{x},z) =
		\lim_{t \to +\infty} \fourier \Delta c_{\sigma}(n_{x},z,t) =
		\lim_{s \to 0} s\, \laplace \fourier \Delta c_{\sigma}(n_{x},z,s)
	\end{equation}
	Separating the limits, according to the following equation,
	aids in the calculation of the Fourier coefficients in steady state
	\begin{multline}
		\fourier \Delta c_{\sigma,f}(n_{x},z) = \\
		\begin{cases}
			\displaystyle
			\pm \lim_{s \to 0} s\, G\!\left(H-z,\, \frac{s}{D_{\sigma}}\right)
			\cdot \lim_{s \to 0} s\, \frac{1}{s}
			\frac{\laplace \fourier \Delta j(0,s)}{F n_{e} D_{\sigma}},
			& n_{x} = 0
			\\[1em]
			\displaystyle
			\pm \lim_{s \to 0} G\!\left(H-z,\, \frac{s}{D_{\sigma}}
			+ n_{x}^{2} \frac{4\pi^{2}}{p_{x}^{2}}\right)
			\cdot \lim_{s \to 0} s\,
			\frac{\laplace \fourier \Delta j(n_{x},s)}{F n_{e} D_{\sigma}},
			& n_{x} \neq 0
		\end{cases}
	\end{multline}
	where the final value of $\fourier \Delta j(n_{x},t)$ is given by
	\begin{equation}
		\fourier \Delta j_{f}(n_{x}) =
		\lim_{t \to +\infty} \fourier \Delta j(n_{x},t) =
		\lim_{s \to 0} s\, \laplace \fourier \Delta j(n_{x},s)
	\end{equation}
	the final value of its time integral is given by
	\begin{equation}
		\int_{0^{-}}^{+\infty} \fourier \Delta j(n_{x},t) \ud{t} =
		\lim_{s \to 0} s\, \frac{1}{s} \laplace \fourier \Delta j(n_{x},s)
	\end{equation}
	and the following limit equals
	\begin{equation}
		\lim_{s \to 0} s\, G\!\left(H-z,\, \frac{s}{D_{\sigma}}\right) =
		\frac{D_{\sigma}}{H}
	\end{equation}
	These lead to the result in steady state
	\begin{equation}
		\fourier \Delta c_{\sigma,f}(n_{x},z) =
		\left\{
			\begin{array}{ll}
				\displaystyle
				\pm \frac{D_{\sigma}}{H} \cdot
				\int_{0^{-}}^{+\infty}
				\frac{\fourier \Delta j(0,t)}{F n_{e} D_{\sigma}}
				\ud{t}, & n_{x} = 0
				\\[1em]
				\displaystyle
				\pm G\!\left(H-z,\, n_{x}^{2} \frac{4\pi^{2}}{p_{x}^{2}}\right)
				\frac{\fourier \Delta j_{f}(n_{x})}{F n_{e} D_{\sigma}}, & n_{x} \neq 0
			\end{array}
		\right.
	\end{equation}

	Therefore, the Fourier coefficients of the full-scale concentrations are obtained by adding Eq. (\ref{coplanar:eqn:ci})
	\begin{equation}
		\fourier c_{\sigma,f}(n_{x},z) =
		\left\{
			\begin{array}{ll}
				\displaystyle
				\bar{c}_{\sigma,i} \pm \frac{1}{H}
				\int_{0^{-}}^{+\infty}
				\frac{\fourier j(0,t)}{F n_{e}}
				\ud{t}, & n_{x} = 0
				\\[1em]
				\displaystyle
				\pm G\!\left(H-z,\, n_{x}^{2} \frac{4\pi^{2}}{p_{x}^{2}}\right)
				\frac{\fourier j_{f}(n_{x})}{F n_{e} D_{\sigma}}, & n_{x} \neq 0
			\end{array}
		\right.
	\end{equation}
	where $\fourier j(0,t) = \fourier j(0,t) - \fourier j_{i}(0) = \fourier \Delta j(0,t)$
	due to Eq. (\ref{coplanar:eqn:kirchhoff:inicial}).
\end{proof}

Considering the previous result,
the first \emph{conservation} property holds for the horizontal average of $c_{\sigma.f}(x,z)$, at any $z$ of the cell,
which may deviate from its initial counterpart
due to unbalanced currents (Kirchhoff's law not satisfied) during the transient state.

\begin{corolario}
	\label{coplanar:cor:bar_cf}
	Assume that the initial concentration $c_{\sigma,i}(x,z)$ comes from a previous steady state
	and the time integral of the net current in Eq. (\ref{coplanar:eqn:i_neta_acumulada}) converges.
	The average of the final concentration, along any horizontal line,
	is independent of $z$ and equals 
	\begin{equation}
		\label{coplanar:eqn:bar_cf}
		\frac{1}{p_{x}} \int_{-p_{x}/2}^{+p_{x}/2} c_{\sigma,f}(x,z) \ud{x} =
		\fourier c_{\sigma,f}(0,z) =
		\bar{c}_{\sigma,f} =
		\bar{c}_{\sigma,i} \pm \Delta \bar{c}_{f}
	\end{equation}
	where $p_{x}$ is one period of the cell
	and $\Delta \bar{c}_{f}$ is independent of the electrochemical species,
	but is proportional to the time integral of the net current,
	as shown in Eq. (\ref{coplanar:eqn:delta_bar_cf}).
\end{corolario}

The second \emph{conservation} property holds for the weighted sum of concentrations
at any point in the cell,
which translates into the total concentration at any point in the cell
when the diffusion coefficients of both species are equal.

\begin{corolario}
	\label{coplanar:cor:cf:total}
	Consider the periodic cell with period $p_{x}$ described in \S\ref{coplanar:def:problema},
	and assume that the initial concentration $c_{\sigma,i}(x,z)$ comes from a previous steady state
	and the time integral of the net current in Eq. (\ref{coplanar:eqn:i_neta_acumulada}) converges. 
	The following weighted sum of the final concentrations is independent of $(x,z)$ and equals
	\begin{equation}
		D_{O} c_{O,f}(x,z) + D_{R} c_{R,f}(x,z) =
		D_{O} \underbrace{(\bar{c}_{O,i} + \Delta \bar{c}_{f})}_{\bar{c}_{O,f}} +
		D_{R} \underbrace{(\bar{c}_{R,i} - \Delta \bar{c}_{f})}_{\bar{c}_{R,f}}
	\end{equation}
	where $\Delta \bar{c}_{f}$ is independent of the electrochemical species,
	but is proportional to the time integral of the net current,
	as shown in Eq. (\ref{coplanar:eqn:delta_bar_cf}).
\end{corolario}

\begin{proof}
	Take the weighted sum of both Fourier coefficients
	\begin{multline}
		D_{O} \fourier c_{O,f}(n_{x},z) + D_{R} \fourier c_{R,f}(n_{x},z)
		\\
		= \begin{cases}
			D_{O} (\bar{c}_{O,i} + \Delta \bar{c}_{f}) +
			D_{R} (\bar{c}_{R,i} - \Delta \bar{c}_{f}), & n_{x} = 0
			\\
			0, & n_{x} \neq 0
		\end{cases}
	\end{multline}
	and later, take its Fourier series.
\end{proof}


\section{Results and discussion}

The results of the theoretical part will be illustrated with a concrete case,
namely, the case of \emph{interdigitated array of electrodes} (IDAE).
For this configuration, it will be seen that the properties of
\emph{horizontal average} and \emph{weighted sum} of concentrations,
together with the physical constraint of \emph{non-negative} concentrations,
impose non-linearities that can affect the limiting current of the cell.
Besides, a rough prediction of the dynamic behavior of the current
and also a prediction of the change in the average concentration on the IDAE is done.
The last two results are contrasted against simulations.

\subsection{Average properties in case of interdigitated arrays}
\label{idae:promedios}

Consider the case of an IDAE configuration
in a cell of height $H$, total width $W_{T}$ and depth $L$,
as shown in Fig. \ref{idae:fig:cell}.
The cell is symmetric along the $y$-axis,
such that a two-dimensional representation $(x,z)$ suffices.
Inside this cell, there are two electrochemical species that react according to Eq. \eqref{coplanar:eqn:reaccion}.

The IDAE consists of two arrays of band electrodes, $A$ (black) and $B$ (gray),
of which two consecutive bands are separated by a center-to-center distance of $W$,
the width of the bands is $2w_{A}$ and $2w_{B}$,
and the number of bands is $N_{A} = N_{B}$ respectively.
The cell may have one of the arrays performing as counter electrode, Fig. \ref{idae:fig:cell:intC},
or have a counter electrode  of width $w_{C}$ external and coplanar to the IDAE, Fig. \ref{idae:fig:cell:extC}.

For the sake of simplicity,
it is assumed that the first and last bands of the IDAE have half width.
Therefore, the IDAE in Fig. \ref{idae:fig:cell:intC} can be represented exactly
as an assembly of units of symmetry of width $W$, height $H$ and half-band electrodes of $A$ and $B$.
Here, each unit of symmetry will be refered to as a \emph{unit cell},
and it is shown in Fig. \ref{idae:fig:cell:unit}.
Similarly, the IDAE in Fig. \ref{idae:fig:cell:extC} can be represented approximately as an assembly of unit cells,
provided that the number of electrode bands $N_{A} = N_{B}$ is sufficiently large,
so that the edge effects at the end of the IDAE are negligible.

\begin{figure}[t]
	\centering
	\subcaptionbox{
		\label{idae:fig:cell:intC}
		Internal counter electrode.}[60mm]{\includegraphics{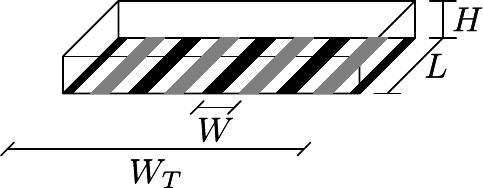}}
	\subcaptionbox{
		\label{idae:fig:cell:extC}
		External counter electrode.}[60mm]{\includegraphics{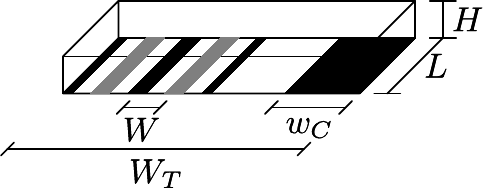}}
	\caption{
		Sketch of interdigitated array of electrodes (IDAE) in a cell of finite height $H$,
		total width $W_{T}$ and depth $L$.
		Fig. (\subref{idae:fig:cell:intC}) can be regarded as an assembly of unit cells of width $W$.
		Fig. (\subref{idae:fig:cell:extC}) can be approximately regarded as an assembly of unit cells of width $W$,
		provided that the number of electrode bands is sufficiently large.
		Note that for both IDAEs, the first and the last bands have half width.
	}
	\label{idae:fig:cell}
\end{figure}

Due to the periodic nature of the IDAE configuration,
the average properties in Corollaries \ref{coplanar:cor:bar_ci}, \ref{coplanar:cor:delta_bar_c} and \ref{coplanar:cor:bar_cf}
must be satisfied at each unit cell of the IDAE ($p_{x} = 2W$)
\begin{subequations}
	\label{idae:eqn:bar_c:unit}
	\begin{align}
		\label{idae:eqn:bar_ci:unit}
		\frac{1}{W} \int_{0}^{W} c_{\sigma,i}(x,z) \ud{x}
		&= \bar{c}_{\sigma,i}\unit
		\\
		\label{idae:eqn:bar_ct:unit}
		\frac{1}{W} \int_{0}^{W} c_{\sigma}(x,z,t) \ud{x}
		&= \bar{c}_{\sigma,i}\unit \pm \Delta\bar{c}_{\sigma}\unit(z,t)
		\\
		\label{idae:eqn:bar_cf:unit}
		\frac{1}{W} \int_{0}^{W} c_{\sigma,f}(x,z) \ud{x}
		&= \underbrace{\bar{c}_{\sigma,i}\unit \pm \Delta\bar{c}_{f}\unit}_{\bar{c}_{\sigma,f}\unit}
	\end{align}
\end{subequations}
either when it fits exactly in the whole cell ($W_{T} = 2W N_{E}$, Fig. \ref{idae:fig:cell:intC}),
or when it doesn't ($W_{T} > 2W N_{E}$, Fig. \ref{idae:fig:cell:extC})
but considering a large number of bands $N_{E}$, with $E \in \{A,B\}$.
Note that the horizontal average in the final steady state holds after a sufficiently long time,
comparable with the time constant of the unit cell.

Besides, the same average properties
in Corollaries \ref{coplanar:cor:bar_ci}, \ref{coplanar:cor:delta_bar_c} and \ref{coplanar:cor:bar_cf}
should also hold for the whole cell ($p_{x} = 2W_{T}$)
\begin{subequations}
	\label{idae:eqn:bar_c:whole}
	\begin{align}
		\label{idae:eqn:bar_ci:whole}
		\frac{1}{W_{T}} \int_{0}^{W_{T}} c_{\sigma,i}(x,z) \ud{x}
		&= \bar{c}_{\sigma,i}\whole
		\\
		\label{idae:eqn:bar_ct:whole}
		\frac{1}{W_{T}} \int_{0}^{W_{T}} c_{\sigma}(x,z,t) \ud{x}
		&= \bar{c}_{\sigma,i}\whole 
		\pm \underbrace{\Delta\bar{c}_{\sigma}\whole(z,t)}_{=0}
		\\
		\label{idae:eqn:bar_cf:whole}
		\frac{1}{W_{T}} \int_{0}^{W_{T}} c_{\sigma,f}(x,z) \ud{x}
		&= \bar{c}_{\sigma,i}\whole 
		\pm \underbrace{\Delta\bar{c}_{f}\whole}_{=0}
	\end{align}
\end{subequations}
since it is surrounded by insulating walls,
and therefore it can be extended periodically along the $x$-axis.
Note that $\Delta\bar{c}_{\sigma}\whole(z,t) = \Delta\bar{c}_{f}\whole = 0$,
since the whole cell always contains its counter electrode,
therefore the its horizontal average remains constant for all $t$.
In particular, for the case of internal counter electrode (Fig. \ref{idae:fig:cell:intC}),
the horizontal average in the whole cell equals that in the unit cell
$\bar{c}_{\sigma,i}\whole = \bar{c}_{\sigma,i}\unit$
and remains unchanged for all $t$.

\subsection{Simulations}

Simulations\footnote{Scripts can be obtained from \url{\elektrodo}.}
were performed for the current in the unit cell using the finite volume PDE solver
\href{http://www.ctcms.nist.gov/fipy/}{FiPy} \cite{Guyer:2009:may}.
The numerical results are compared with their theoretical counterparts in the coming sections.

For the sake of simplicity, it is assumed that
the charge transfer on the electrodes follows reversible electrode reactions.
Also it is assumed that the species have equal diffusion coefficients $D_{O} = D_{R} = D$.
Both assumptions mean that the Nernst equation on both electrodes can be decoupled not only in steady state, but also during the transient, due to Corollary \ref{coplanar:cor:ct:total}.
This allows simulating the concentration of each electrochemical species independently.

The simulations consist of a normalized diffusion equation for the unit cell of Fig. \ref{idae:fig:cell:unit}
\begin{subequations}
	\label{idae:eqn:pde:idae:normalized}
	\begin{align}
		\pi^{2} \parderiv{\xi\simu}{t\simu}(x, z, t) 
		&= \parderiv{^2 \xi\simu}{x\simu^{2}}(x, z, t)
		+ \parderiv{^{2} \xi\simu}{z\simu^{2}}(x, z, t) \\
		\parderiv{\xi\simu}{x\simu}(0,z,t)
		&= \parderiv{\xi\simu}{x\simu}(W, z, t) = 0,
		\, \forall z \in [0,H]
	\end{align}
	\vspace{-1.5\baselineskip}
	\begin{gather}
		\parderiv{\xi\simu}{z\simu}(x, H, t) = 0,\, \forall x \in [0,W],
		\: \parderiv{\xi\simu}{z\simu}(x,0,t) = 0,
		\, \forall x \notin A \cup B
		\\
		\xi\simu(x,0,t) = 0,\, \forall x \in A,
		\quad \xi\simu(x,0,t) = 1,\, \forall x \in B
		\label{idae:eqn:pde:idae:normalized:onAB}
	\end{gather}
\end{subequations}
where
\begin{subequations}
	\label{idae:eqn:pde:normalization}
	\begin{gather}
		\xi\simu(x,z,t) = \frac{c_{\sigma}(x,z,t) - c_{\sigma,f}^{A}}{c_{\sigma,f}^{B} - c_{\sigma,f}^{A}}
		\\
		x\simu = \frac{x}{W}
		,\quad z\simu = \frac{z}{W}
		,\quad t\simu = \frac{\pi^{2} D t}{W^{2}}
	\end{gather}
\end{subequations}
which considers the transition  
from two possible initial states $\xi\simu(x,z,0^{-}) \in \{\numlist[list-pair-separator={,}]{0,25; 0,5}\}$ to its final state $\xi\simu(x,z,+\infty)$.

The width of each band electrode was taken equal to $2w_{A} = 2w_{B} = \num{0,5} W$ for all simulations
and three aspect ratios for the unit cell were considered
$H/W \in \{\numlist[list-final-separator={,}]{0,3; 0,5; 1,0}\}$.

\begin{figure}[t]
	\centering
	\subcaptionbox{
		Mesh: $n_{x} \times n_{z} = 124 \times 37$ and $\delta_{0} = \num{6,25e-5}$.
	}{\includegraphics{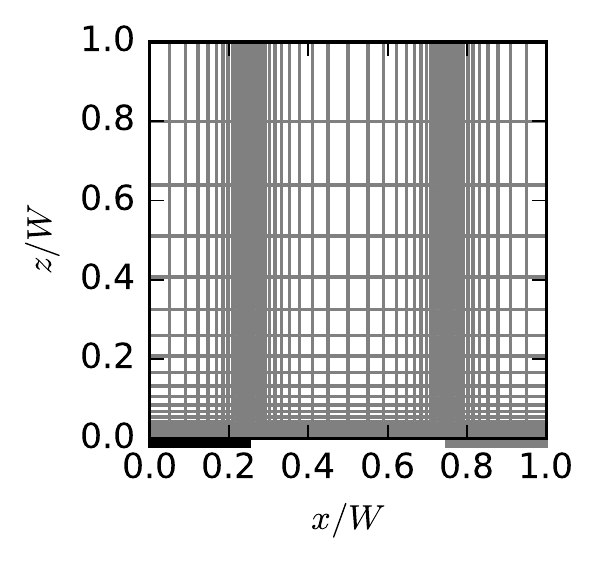}}
	\hspace{1em}
	\subcaptionbox{
		Mesh: $n_{x} \times n_{z} = 124 \times 34$ and $\delta_{0} = \num{6,25e-5}$.
	}{\includegraphics{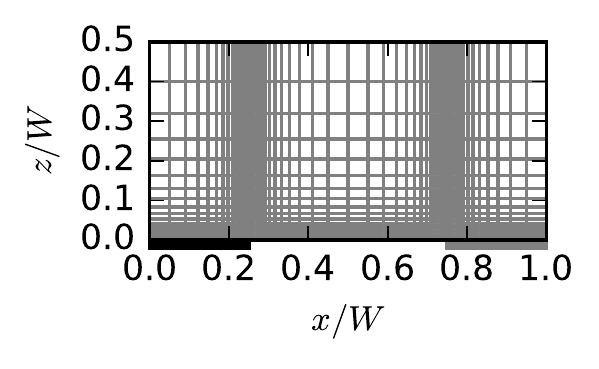}}
	\subcaptionbox{
		Mesh: $n_{x} \times n_{z} = 120 \times 31$ and $\delta_{0} = \num{7,5e-5}$.
	}{\includegraphics{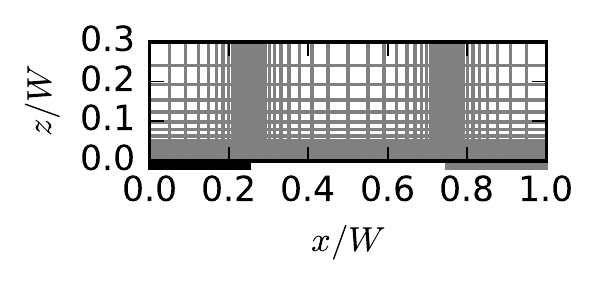}}
	\caption{
		Exponential meshes used in the simulations.
		The size of each mesh is $n_{x} \times n_{z}$
		and the dimensions of its smallest element are $\delta_{x} = \delta_{z} = \delta_{0}$.
	}
	\label{idae:fig:mesh}
\end{figure}

An exponential mesh was used to partition the unit cell \cite[\S7.2]{Britz:2016:},
in order to keep the memory usage low while maintaining good resolution near the electrode bands, see Fig. \ref{idae:fig:mesh}.
The number of elements of the mesh is $n_{x} \times n_{z}$,
of which the width and height of its smallest element are $\delta_{x} = \delta_{z} = \delta_{0}$.
The mesh was succesively refined until the absolute error of the current in steady state,
between two consecutive refinements, was less than \num{0,5e-4}
(which corresponds approximately to four decimal places of agreement between refinemts).
See Suplementary Information \S\ref{data:s32:mesh} for the output of the script of mesh refinement.

Fig. \ref{idae:fig:i_sim} shows the simulated current $i\simu^{E/2}(t)$
through a half-band electrode of $E \in \{A,B\}$
\begin{equation}
	i\simu^{E/2}(t)
	= \int_{E/2} -\parderiv{\xi\simu}{z\simu}(x,0,t) \ud{x\simu}
	= \int_{E/2} -\parderiv{\xi\simu}{z}(x,0,t) \ud{x}
\end{equation}
which was obtained by numerically solving Eqs. (\ref{idae:eqn:pde:idae:normalized})
subject to the initial condition $\xi\simu(x,z,0^{-}) \in \{\numlist[list-pair-separator={,}]{0,25; 0,5}\}$.

\begin{figure*}
	\centering
	\subcaptionbox{
		$H/W = \num{1}$.
		$|i\simu^{E/2}(+\infty)| \approx \num{0,496}$.
		$\Delta\bar{c}_{f}/[c_{\sigma,f}^{B} - c_{\sigma,f}^{A}]
		\approx \{\num{\pm 0,249},\, 0\}$.
	}{
		\includegraphics{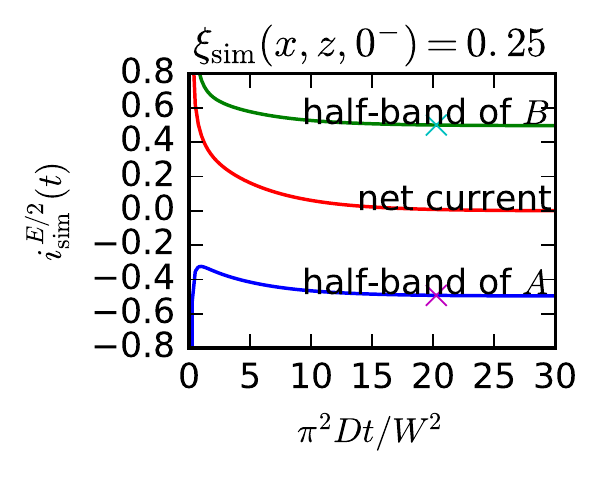} \quad
		\includegraphics{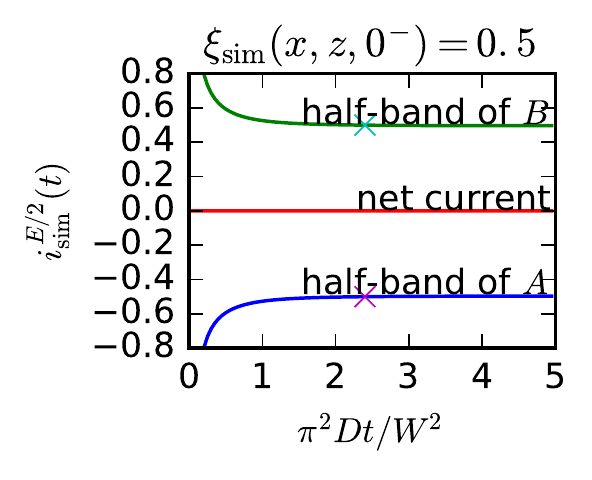}
	}

	\subcaptionbox{
		$H/W = \num{0,5}$.
		$|i\simu^{E/2}(+\infty)| \approx \num{0,460}$.
		$\Delta\bar{c}_{f}/[c_{\sigma,f}^{B} - c_{\sigma,f}^{A}]
		\approx \{\num{\pm 0,250},\, 0\}$.
	}{
		\includegraphics{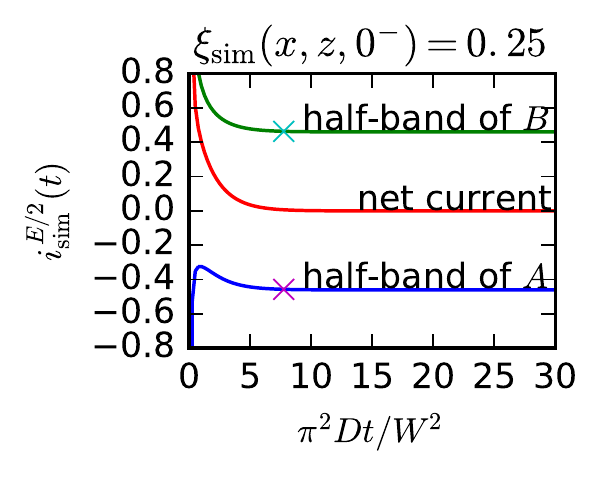} \quad
		\includegraphics{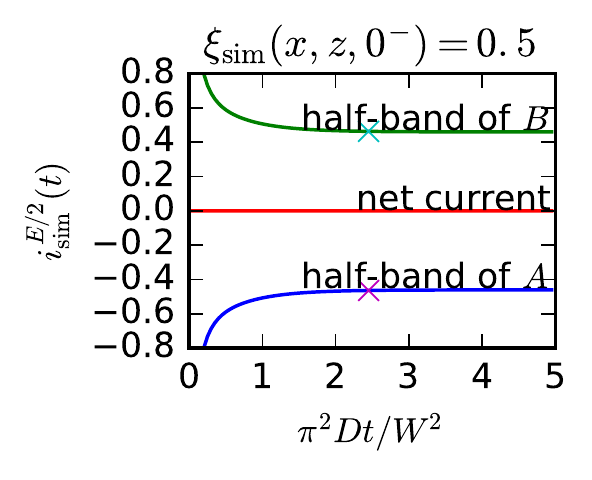}
	}

	\subcaptionbox{
		$H/W = \num{0,3}$.
		$|i\simu^{E/2}(+\infty)| \approx \num{0,377}$.
		$\Delta\bar{c}_{f}/[c_{\sigma,f}^{B} - c_{\sigma,f}^{A}]
		\approx \{\num{\pm 0,250},\, 0\}$.
	}{
		\includegraphics{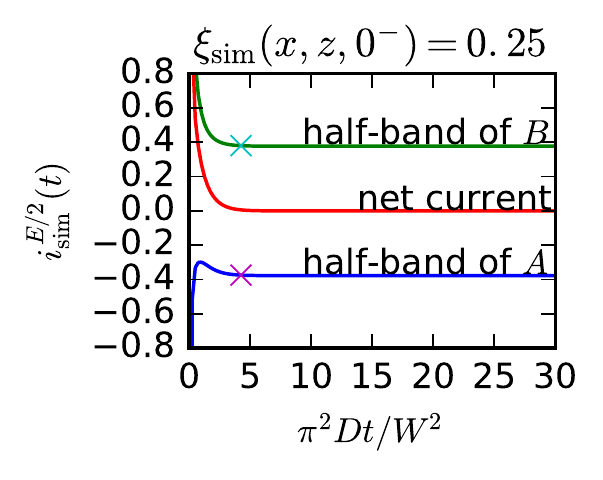} \quad
		\includegraphics{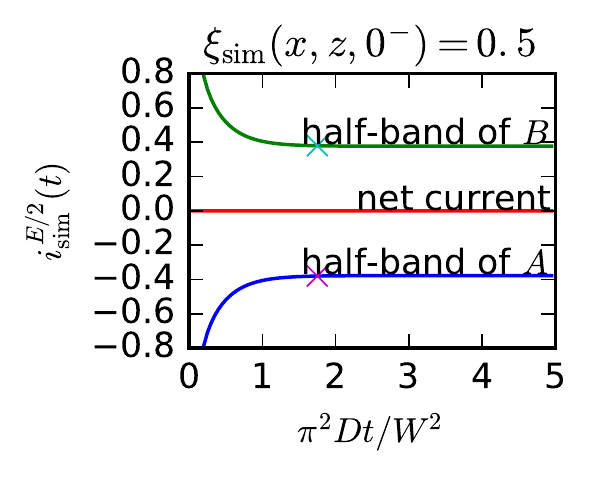}
	}
	\caption{
		Simulated current $i\simu^{E/2}(t) = \int_{E/2} -\partial \xi\simu(x,0,t)/ \partial z\simu \ud{x\simu}$ through a half-band electrode of $E \in \{A,B\}$,
		as a function of time, for different aspect ratios of the unit cell $H/W$ and initial concentrations $\xi\simu(x,z,0^{-})$.
		The `$\times$' show the time required by the simulated current to reach \SI{0,67}{\percent} of its steady-state value.
		Left column:
		Average concentrations in initial and final steady states are different $\xi\simu(x,z,0^{-}) \neq 1/2$.
		Right column:
		Average concentrations in initial and final steady states are equal $\xi\simu(x,z,0^{-}) = 1/2$.
		The simulations were obtained by numerically solving Eqs. (\ref{idae:eqn:pde:idae:normalized}).
	}
	\label{idae:fig:i_sim}
\end{figure*}

\subsection{Effect of the counter electrode on the net current}

The fact of having an IDAE with internal or external counter electrode
influences the time that its current requires to reach steady state
and also its collection efficiency.
Both effects can be obtained as consecuence of the average properties of
Corollaries \ref{coplanar:cor:delta_bar_c} and \ref{coplanar:cor:bar_cf},
and will be discussed below.

\subsubsection{Time to reach steady state}

When using an external counter electrode (both arrays are potentiostated),
the average concentration at the bottom of the unit cell ($z = 0$)
is in general forced to a value different than its initial counterpart $\bar{c}_{\sigma,i}\unit$.
In case of the simulation in Eqs. \eqref{idae:eqn:pde:idae:normalized} and \eqref{idae:eqn:pde:normalization},
this equals the arithmetic average of the concentrations on both electrodes,
due to symmetry, since the electrode bands have equal width $2w_{A} = 2w_{B}$
\begin{equation}
	\label{idae:eqn:bar_x_sim}
	\frac{1}{W} \int_{0}^{W} \xi\simu(x,0,t) \ud{x} = \frac{1}{2}
	\Leftrightarrow
	\frac{1}{W} \int_{0}^{W} c_{\sigma}(x,0,t) \ud{x}
	= \frac{c_{\sigma,f}^{A} + c_{\sigma,f}^{B}}{2}
\end{equation}
which is different from its initial counterpart when $\xi\simu(x,z,0^{-}) = \num{0,25}$.
This produces a change of $\pm \Delta\bar{c}_{\sigma}\unit(0,t) \neq 0$
in the average concentration at the bottom of the unit cell,
see Eq. \eqref{idae:eqn:bar_ct:unit},
which subsequently generates a non-zero net current in the unit cell during transient state,
due to Eq. \eqref{coplanar:eqn:i_net} with $p_{x} = 2W$.

The simulations at the left column of Fig. \ref{idae:fig:i_sim} show
the generation of a net current in the unit cell when $\xi\simu(x,z,0^{-}) = \num{0,25}$.
Note that the net current in the unit cell,
as well as the current at each half-band electrode,
have similar dynamics and reach steady state nearly at the same time.
This time can be predicted from Eq. (\ref{coplanar:eqn:i_net}) in Corollary \ref{coplanar:cor:delta_bar_c},
since the net current in the unit cell has natural modes of the form
\begin{equation}
	\label{idae:eqn:tau:extC}
	\exp\!\left( -(2\ell - 1)^{2} \frac{W^{2}}{4H^{2}} \cdot \frac{\pi^{2}}{W^{2}} D_{\sigma} t \right),\quad \ell = 1,2,\ldots
\end{equation}
from its impulse response $F n_{e} D_{\sigma}^{2}\, h(H,D_{\sigma}t)$, which decay exponentially with time.
The slowest of these exponential modes, that is with $\ell=1$,
is the one that gives an idea of the time required to reach steady state.
This time is roughly approached when $\pi^{2} D_{\sigma} t/W^{2} = 5 \cdot (2H/W)^{2}$,
that is when the slowest exponential mode approximately vanishes $\exp(-5) \approx \SI{0,67}{\percent}$.
The left column of Fig. \ref{idae:fig:i_sim} shows with `$\times$'
the times needed for the simulated current to reach \SI{0,67}{\percent} of its steady-state value,
which correspond roughly to their theoretical counterparts: 20, 5 and \num{1,8}.

On the other hand, when one of the arrays performs as counter electrode (internal counter),
the net current in the unit cell must remain always zero.
This fact suggests that the average concentration at the bottom of the unit cell
is forced by the potentiostat to its initial counterpart $\bar{c}_{\sigma,i}\unit$.
In case of the simulation,
this average is forced to $\xi\simu(x,z,0^{-}) = \num{0,5}$.
This produces no change in average at the bottom of the unit cell
$\Delta \bar{c}_{\sigma}\unit(0,t) = 0$,
which is the cause of having zero net current during the transient.

The simulations at the right column of Fig. \ref{idae:fig:i_sim} show zero net current when $\xi\simu(x,z,0^{-}) = \num{0,5}$.
Despite the net current in the unit cell is zero,
the current at each array does evolve with time,
reaching its steady state in a shorter time than in the case of external counter
(compare with left column of Fig. \ref{idae:fig:i_sim}).
This behavior can be explained by looking at the Fourier series of the current density
\begin{equation}
	j(x,t) = \sum_{n_{x}=-\infty}^{+\infty} \fourier j(n_{x},t) \e^{\bm{i} x\, n_{x} \pi/W}
\end{equation}
where $\fourier j(n_{x},t)$ correspond to its Fourier coefficients.
Note that the Fourier coefficient $\fourier j(n_{x},t)$ with $n_{x}=0$
corresponds to the average component of the current density (net current) in the unit cell,
which equals zero when the counter electrode is internal to the IDAE.
Therefore, only the Fourier coefficients $\fourier j(n_{x},t)$ with $n_{x} \neq 0$ vary with time,
and they do so according to the impulse response
\begin{equation}
	\label{idae:eqn:h}
	\laplace^{-1} G\!\left( H, \frac{s}{D_{\sigma}} + n_{x}^{2} \frac{\pi^{2}}{W^{2}} \right)^{-1}
	=
	D_{\sigma}\,
	h(H,D_{\sigma} t) \exp\!\left( -n_{x}^{2} \frac{\pi^{2}}{W^{2}} D_{\sigma} t \right)
\end{equation}
from Eq. (\ref{coplanar:eqn:Dc}) in Theorem \ref{coplanar:teo:Dc}
and Eq. \eqref{coplanar:eqn:h} in Lemma \ref{coplanar:lem:g-h}.
Thus, the current density exhibits exponential modes that decay with time according to
\begin{equation}
	\label{idae:eqn:tau:intC}
	\exp\!\left( -\left[ n_{x}^{2} + (2\ell - 1)^{2} \frac{W^{2}}{4H^{2}} \right] \frac{\pi^{2}}{W^{2}} D_{\sigma} t \right),\quad 
	\begin{array}{rcl}
		n_{x} &=& \pm 1, \pm 2, \ldots \\
		\ell &=& 1, 2, \ldots
	\end{array}
\end{equation}
From all these exponential modes,
it is the slowest, that is with $n_{x} = \pm 1$ and $\ell=1$,
the one that gives an idea of the time required to reach steady state.
This time is roughly approached when $\pi^{2} D_{\sigma} t/W^{2} = 5 \cdot [1 + (\num{0,5} W/H)^{2}]^{-1}$,
that is when the slowest exponential mode approximately vanishes $\exp(-5) \approx \SI{0,67}{\percent}$.
At the right column of Fig. \ref{idae:fig:i_sim},
the times required by the simulated current
to reach \SI{0,67}{\percent} of its steady-state value are shown with `$\times$'
and correspond roughly to their theoretical counterparts: 4, \num{2,5} and \num{1,3}.

Finally, and independently of using internal or external coun\-ter electrode,
the time response of the current tends to speed up
as the height of the cell $H$ decreases.
This is justified by the shorter distances
that the electrochemical species must travel, due to lower roof of the cell.

\subsubsection{Collection efficiency in steady state}

For finite cell height $H$,
the steady-state current through a pair of electrode bands $A$ and $B$ is equal ($i_{f}^{A} = -i_{f}^{B}$).
This is confirmed by Eq. \eqref{coplanar:eqn:kirchhoff:final} with $p_{x} = 2W$ and Eq. \eqref{idae:eqn:if},
and it is shown in all plots of Fig. \ref{idae:fig:i_sim} after a sufficiently long time.
Therefore, 100\% collection efficiency must be obtained inside a unit cell,
independently of whether the counter electrode is internal or external.

But for cell heights approaching infinite $H \to +\infty$,
the collection efficiency is different for internal and external counter electrodes.
If the counter electrode is internal (one array performs as counter),
then the collection efficiency in the unit cell is automatically 100\%.
However, if the counter electrode is external,
then the collection efficiency is less than 100\%
when the average of the final concentration at the bottom of the unit cell
is forced to a different value than $\bar{c}_{\sigma,i}\unit$.

Collection efficiencies lower than 100\% in steady state,
for external counter electrode and very tall cells $H \to +\infty$,
can be explained by recalling the change in average concentration
at the bottom of a unit cell ($z=0$).
See Eqs. \eqref{coplanar:eqn:delta_bar_c} and \eqref{coplanar:eqn:delta_bar_cf}
\begin{equation}
	\label{idae:eqn:delta_bar_cf}
	\Delta\bar{c}_{\sigma}\unit(z, +\infty)
	= \Delta\bar{c}_{f}\unit
	= \frac{1}{H W L}
	\int_{0^{-}}^{+\infty}
		\frac{1}{F n_{e}}
		\underbrace{
			\int_{0}^{W} j(x,t)\, L \ud{x}
		}_{i^{\text{net}}(t) \text{ in unit cell}}
	\ud{t}
\end{equation}
Since fixing the average concentration at $z = 0$
to a value different than $\bar{c}_{\sigma,i}\unit$ means that 
$\Delta\bar{c}_{\sigma}\unit(0, +\infty) \neq 0$ is fixed to a finite value,
then the time integral $|\int_{0^{-}}^{+\infty} i^{\text{net}}(t) \ud{t}| \to +\infty$
is forced to diverge when $1/H \to 0^{+}$.
The infinite value of this integral is obtained when $i^{\text{net}}(+\infty) \neq 0$,
leading to a collection efficiency that is different from 100\% in steady state.

In this last case,
Corollary \ref{coplanar:cor:bar_cf} breaks due to $i^{\text{net}}(+\infty) \neq 0$,
producing a horizontal average of concentration, locally over the IDAE,
that is not uniform along the $z$-axis.
Therefore, a correction that takes into account
the effect of an external counter electrode (both arrays are potentiostated)
is needed to accurately predict the steady-state current through the IDAE.
This kind of correction was done the semi-empirically in \cite[Eq. (33)]{Aoki:1988:dec}
and later in \cite[Eqs. (13) and (20)]{Morf:2006:may}, both for the case of semi-infinite cells ($H \to +\infty$).

\subsection{Effect of net current on the average concentration}

The net current entering the unit cell
plays a determinant role on the horizontal average of concentration
for the entire unit cell at steady state.

As seen in the previous sections, a change of average concentration
at the bottom of the unit cell $\Delta\bar{c}_{\sigma}\unit(0,t)$
produces a non-zero net current due to Eq. \eqref{coplanar:eqn:i_net}.
Subsequently, this net current produces a change in horizontal average of concentration
at the entire unit cell $\Delta\bar{c}_{\sigma}\unit(z,t)$, due to Eq. \eqref{coplanar:eqn:delta_bar_c},
which reaches a steady state $\Delta\bar{c}_{f}\unit$
that is uniform $\forall z$ and independent of the electrochemical species $\sigma$, see Eq. \eqref{idae:eqn:delta_bar_cf}.

Therefore, the horizontal averages at the entire unit cell
for the final and initial steady sates are, in general, different
($\bar{c}_{\sigma,f}\unit \neq \bar{c}_{\sigma,i}\unit$) and this difference
($\bar{c}_{\sigma,f}\unit = \bar{c}_{\sigma,i}\unit \pm \Delta \bar{c}_{f}\unit$)
depends on the net current during the transition
from the initial towards the final state,
as seen in Eq. (\ref{idae:eqn:delta_bar_cf}).

If the net current is different from zero
during some finite time interval,
the currents through the generator and collector are different
and accumulation (or depletion) of species occurs inside the unit cell.
This generates the deviation of $\bar{c}_{\sigma,f}\unit$ with respect to $\bar{c}_{\sigma,i}\unit$.
Conversely, if the net current is zero for all $t$,
the currents at the generator and collector are equal
and no accumulation (or depletion) of species occurs.

In case of the simulation in Eq. \eqref{idae:eqn:pde:idae:normalized},
$\bar{c}_{\sigma,f}\unit$ must be given by 
the arithmetic average on both electrodes, since $2w_{A} = 2w_{B}$
\begin{equation}
	\frac{1}{W} \int_{0}^{W} c_{\sigma,f}(x,z) \ud{x}
	= \underbrace{\frac{c_{\sigma,f}^{A} + c_{\sigma,f}^{B}}{2} }_{ \bar{c}_{\sigma,f}\unit }
	\Leftrightarrow
	\frac{1}{W} \int_{0}^{W} \xi\simu(x,z,+\infty) \ud{x} = \frac{1}{2}
\end{equation}
The expression
$\bar{c}_{\sigma,f}\unit = \bar{c}_{\sigma,i}\unit \pm \Delta \bar{c}_{f}\unit$
has its simulated counterpart given by
\begin{equation}
	\bar{c}_{\sigma,f}\unit
	= \bar{c}_{\sigma,i}\unit \pm \Delta \bar{c}_{f}\unit
	\Leftrightarrow
	\frac{1}{2} 
	= \xi\simu(x,z,0^{-})
	\pm \frac{\Delta\bar{c}_{f}\unit}{[c_{\sigma,f}^{B} - c_{\sigma,f}^{A}]}
\end{equation}
meaning that the final average for the simulation must be $1/2$
independently of the initial concentration $\xi\simu(x,z,0^{-})$.
Also the change in horizontal average from Eq. \eqref{idae:eqn:delta_bar_cf}
was normalized to obtain
\begin{equation}
	\label{idae:eqn:delta_bar_cf_sim}
	\frac{\Delta\bar{c}_{f}\unit}{[c_{\sigma,f}^{B} - c_{\sigma,f}^{A}]}
	= \pm \frac{1}{\pi^{2}} \frac{W}{H}
	\int_{0^{-}}^{+\infty}
		\underbrace{
			\int_{0}^{1}
				-\parderiv{\xi\simu}{z\simu}(x,0,t)
			\ud{x\simu}
		}_{
			i\simu^{\text{net}}(t) 
			= i\simu^{A/2}(t) + i\simu^{B/2}(t)
		}
	\ud{t\simu}
\end{equation}
All normalizations were obtained by applying Eq. \eqref{idae:eqn:pde:normalization}

Fig \ref{idae:fig:i_sim} shows that, when $\xi\simu(x,z,0^{-}) = \num{0,25}$,
the horizontal average reaches $1/2$ in the final state, due to $i\simu^{\text{net}}(t) \geq 0$.
In this case, the change in horizontal average was obtained numerically
by computing the time integral of $i\simu^{\text{net}}(t)$ in Eq. \eqref{idae:eqn:delta_bar_cf_sim},
and approaches its theoretical value
$\pm \Delta\bar{c}_{f}\unit/[c_{\sigma,f}^{B} - c_{\sigma,f}^{A}] = \num{0,25}$
up to two decimal places for all simulated aspect ratios $H/W$
(see also Supplementary Information \S\ref{data:s32:simulation} for full numerical values).
On the other hand, when $\xi\simu(x,z,0^{-}) = \num{0,5}$,
the simulated net current $i\simu^{\text{net}}(t)$ equals zero for all $t$.
This produces no change in horizontal average
$\Delta\bar{c}_{f}\unit/[c_{\sigma,f}^{B} - c_{\sigma,f}^{A}] = 0$,
such that it can be maintained at $1/2$  until the final steady state ($\bar{c}_{\sigma,f}\unit = \bar{c}_{\sigma,i}\unit$).

Despite of not being mentioned explicitly in the literature,
earlier results showing change in average (bulk) concentration, due to non-zero net currents,
can also be found in the simulations of \cite[Fig. 5 and Eq. (4) with boundary conditions for coplanar electrodes]{Strutwolf:2005:feb}.

\subsection{Constraints on the limiting current}

\begin{figure}[t]
	\centering
	\subcaptionbox{
		\label{idae:fig:cell:unit}
		IDAE unit cell.}{\includegraphics{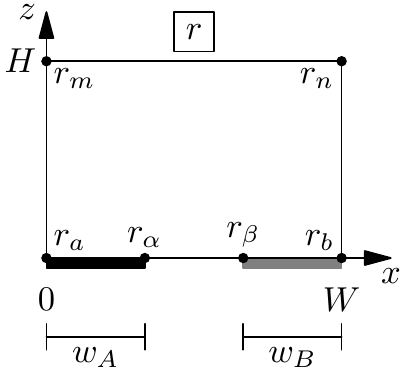}}
	\subcaptionbox{
		\label{idae:fig:cell:parallel}
		Transformed unit cell.}[15em]{\includegraphics{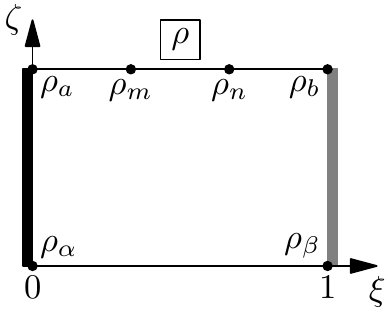}}  
	\caption{
		Complex transformation $\bm{\rho} = T(\bm{r})$ of the unit cell
		from IDAE domain $\bm{r} = (x,z)$
		to parallel plates domain $\bm{\rho} = (\xi, \zeta)$.
		The boundary points $\bm{r}_{p}$ are bijectively mapped to $\bm{\rho}_{p}$,
		where $p \in \{a, \alpha, \beta, b, m, n\}$.
		This transformation can be obtained following a similar process to the one stated in \cite{Aoki:1988:dec},
		and its conformality is ensured by the conformality of Möbius functions \cite[\S5.7]{Ablowitz:2003:apr}
		and by the conformality of \emph{Schwarz-Christoffel} transformations \cite[Theorem 5.6.1]{Ablowitz:2003:apr}.
	}
	\label{idae:fig:cell:transformation}
\end{figure}

Before presenting the results on constraints for the limiting current,
it is convenient to show that, under certain conditions,
the current in steady state through the IDAE is proportional to the difference of concentration on both arrays.

\begin{lema}
	\label{idae:lem:if}
	Consider an IDAE electrochemical cell under the assumptions of \S\ref{idae:promedios}.
	If the final concentration of species $\sigma \in \{O,R\}$ at each electrode band $E \in \{A,B\}$ is uniform and equal to $c_{\sigma,f}^{E}$,
	then the current in the final steady state $i_{f}^{E}$
	through a band $E \in \{A,B\}$ is proportional to the difference $[c_{\sigma,f}^{E} - c_{\sigma,f}^{E'}]$
	\begin{equation}
		\label{idae:eqn:if}
		\pm \frac{
			i_{f}^{A}/L
		}{
			F n_{e} D_{\sigma} [c_{\sigma,f}^{A} - c_{\sigma,f}^{B}]
		}
		= \pm \frac{
			i_{f}^{B}/L
		}{
			F n_{e} D_{\sigma} [c_{\sigma,f}^{B} - c_{\sigma,f}^{A}]
		}
		= 2 \zeta(0,0)
	\end{equation}
	where $E'$ is the complementary band of $E$,
	and $\zeta(x,z)$ corresponds to the imaginary part of the conformal transformation $(\xi,\zeta) = T(x,z)$ shown in Fig.\ref{idae:fig:cell:transformation}.
\end{lema}

Note that this result is valid in steady state, that is,
after a sufficiently long time,
comparable with the time constant of the unit cell.
See Eqs. \eqref{idae:eqn:tau:extC} and \eqref{idae:eqn:tau:intC}
for external and internal counter electrode respectively.

\begin{proof}
	Consider the parallel-plates cell in Fig. \ref{idae:fig:cell:parallel}.
	It is known that the concentration profile $\gamma_{\sigma,f}(\xi,\zeta)$
	of species $\sigma \in \{O,R\}$ in the final steady state
	is given by a linear interpolation of the concentration at its electrodes
	\begin{subequations}
		\begin{align}
		\gamma_{\sigma,f}(\xi,\zeta)
		&= c_{\sigma,f}^{A} + [c_{\sigma,f}^{B} - c_{\sigma,f}^{A}] \xi
		\\
		\intertext{
			By returning to the IDAE domain
			$c_{\sigma,f}(x,z) = \gamma_{\sigma,f}(\xi,\zeta)$
			through the domain transformation $(\xi,\zeta) = T(x,z)$
			one obtains
		}
		c_{\sigma,f}(x,z)
		&= c_{\sigma,f}^{A} + [c_{\sigma,f}^{B} - c_{\sigma,f}^{A}]\xi(x,z)
		\end{align}
	\end{subequations}
	where $\xi(x,z)$ corresponds to the real part of the conformal transformation $(\xi,\zeta) = T(x,z)$.
	
	The current in final steady state $i_{f}^{E}$ can be obtained
	by integrating the flux through one electrode band $E \in \{A,B\}$
	\begin{equation}
		i_{f}^{E}
		= \mp \int_{E} F n_{e}\, D_{\sigma} \parderiv{c_{\sigma,f}}{z}(x,0) L \ud{x}
		= \mp \int_{E} F n_{e}\, D_{\sigma} [c_{\sigma,f}^{B} - c_{\sigma,f}^{A}] \parderiv{\xi}{z}(x,0) L \ud{x}
	\end{equation}
	Using the Cauchy-Riemann identities \cite[Theorem 3.2]{Olver:2017:} for $\bm{\rho} = T(\bm{r})$
	\begin{equation}
		\parderiv{\xi}{z} = -\parderiv{\zeta}{x} = -\Im \parderiv{\bm{\rho}}{x}
	\end{equation}
	the current can be further simplified
	\begin{equation}
		\label{idae:eqn:int_dxi7dz_dx}
		\mp \frac{
			i_{f}^{E}/L
		}{
			F n_{e} D_{\sigma} [c_{\sigma,f}^{B} - c_{\sigma,f}^{A}]
		}
		= \int_{E} \parderiv{\xi}{z}(x,0) \ud{x}
		= -\Im \int_{E} \parderiv{\bm{\rho}}{x}(x,0) \ud{x}
	\end{equation}
	Due to symmetry, this integral can be taken in half electrode band
	\begin{subequations}
		\begin{align}
			-\Im \int_{A} \partial\bm{\rho}(\bm{r})
			&= -2\,\Im\bm{\rho}(\bm{r})\Big|_{\bm{r}_{a}}^{\bm{r}_{\alpha}}
			= +2\, \Im \bm{\rho}(\bm{r}_{a})
			= +2\, \zeta(0,0)
			\\
			-\Im \int_{B} \partial\bm{\rho}(\bm{r})
			&= -2\,\Im\bm{\rho}(\bm{r})\Big|_{\bm{r}_{\beta}}^{\bm{r}_{b}}
			= -2\, \Im \bm{\rho}(\bm{r}_{b})
			= -2\, \zeta(0,0)
		\end{align}
	\end{subequations}
	where the imaginary parts $\Im\bm{\rho}_{\beta} = \Im\bm{\rho}_{\alpha} = 0$
	and $\Im\bm{\rho}_{b} = \Im\bm{\rho}_{a} = \zeta(0,0)$
	lead to the result in Eq. (\ref{idae:eqn:if}).
\end{proof}

Once it is clear that the current $i_{f}^{E}$ is proportional to
the difference of concentration between both arrays $[c_{\sigma,f}^{E} - c_{\sigma,f}^{E'}]$,
then it can be shown that non-negative concentrations,
together with the properties of \emph{horizontal average} and \emph{weighted sum} of concentrations,
restrict the maximum current that the cell can produce by directly limiting the difference $[c_{\sigma,f}^{E} - c_{\sigma,f}^{E'}]$.

\begin{teorema}
	\label{idae:teo:cE-cE':intC}
	Consider an IDAE electrochemical cell under the assumptions of \S\ref{idae:promedios}.
	Assume also that the array of bands $E \in \{A,B\}$ and its complementary array of bands $E'$
	perform as working and counter electrodes respectively,
	see Fig. \ref{idae:fig:cell:intC}.

	If the concentrations of species $\sigma \in \{O,R\}$ on both
	arrays are uniform and equal to $c_{\sigma}^{E}(t)$ and $c_{\sigma}^{E'}(t)$,
	and the bands have equal width $2w_{A} = 2w_{B}$, 
	then the concentrations on the working and counter electrodes
	and their difference are related for all $t$
	\begin{equation}
		\label{idae:eqn:cE':intC}
		2[c_{\sigma}^{E}(t) - \bar{c}_{\sigma,i}\whole]
		= -2[c_{\sigma}^{E'}(t) - \bar{c}_{\sigma,i}\whole]
		= [c_{\sigma}^{E}(t) - c_{\sigma}^{E'}(t)]
	\end{equation}
	where $\bar{c}_{\sigma,i}\whole=\bar{c}_{\sigma,i}\unit$
	are the horizontal averages in the initial steady state,
	defined in Eqs. \eqref{idae:eqn:bar_ci:whole} and \eqref{idae:eqn:bar_ci:unit}.
	
	Moreover, the difference of concentrations in steady state
	is limited from above and below by
	\begin{equation}
		\label{idae:eqn:cotas:intC}
		-2 D_{\lambda} \bar{c}_{\lambda,i}\whole
		\leq D_{\sigma} [c_{\sigma,f}^{E} - c_{\sigma,f}^{E'}]
		\leq 2 D_{\lambda} \bar{c}_{\lambda,i}\whole
	\end{equation}
	where the determinant species $\lambda \in \{O,R\}$ is such that
	$D_{\lambda} \bar{c}_{\lambda,i}\whole = \min(D_{O} \bar{c}_{O,i}\whole, D_{R} \bar{c}_{R,i}\whole)$,
	$c_{\sigma,f}^{E} = c_{\sigma}^{E}(+\infty)$,
	and $c_{\sigma,f}^{E'} = c_{\sigma}^{E'}(+\infty)$.
	This last expression determines the limiting current of the cell.
\end{teorema}

It is important to note that Eq. \eqref{idae:eqn:cotas:intC}
is valid in steady state, that is, after a sufficiently long time,
comparable with the time constant of the unit cell,
see Eq. \eqref{idae:eqn:tau:intC} for internal counter electrode.
However, this expression can hold also $\forall t$
under the additional restriction $D_{O} = D_{R}$
(when applying Corollary \ref{coplanar:cor:ct:total}
instead of Corollary \ref{coplanar:cor:cf:total}
in Eqs. \eqref{idae:eqn:cf:unit:total}).

Note also that Eq. (\ref{idae:eqn:cotas:intC}) imply the necessity of
the simultaneous presence of $\bar{c}_{O,i}\whole$ and $\bar{c}_{R,i}\whole$
in order to achieve steady state currents.
This was also noted by \cite[before and after Eq. (17)]{Morf:2006:may}
and explained in \cite[\S2.3 and Fig. 2]{GuajardoYevenes:2013:sep},
both for the case when $D_{O} = D_{R}$.

\begin{proof}
	Consider the properties of horizontal averages in Eqs. \eqref{idae:eqn:bar_c:unit}.
	Since the counter electrode is internal to the IDAE,
	then the net current in the unit cell is zero for all $t$,
	therefore $\Delta \bar{c}_{\sigma,f}\unit(z,t) = \Delta \bar{c}_{f}\unit = 0$.
	Also since $\bar{c}_{\sigma,i}\unit = \bar{c}_{\sigma,i}\whole$ for Fig. \ref{idae:fig:cell:intC},
	then the average properties can be summarized for all $t$ as
	\begin{subequations}
		\begin{align}
			\int_{0}^{W} c_{\sigma}(x,z,t) - \bar{c}_{\sigma,i}\whole \ud{x} &= 0
			\\
			\intertext{%
				Due to bands of equal widths $2w_{A} = 2w_{B}$,
				the integral in the gap between consecutive bands
				$\int_{w_{A}}^{W - w_{B}} c_{\sigma}(x,z,t) - \bar{c}_{\sigma,i}\whole \ud{x}$
				 equals zero for all $t$.
				Therefore, the average property is reduced to
			}
			\label{idae:eqn:average:intC}
			[c_{\sigma}^{E}(t) - \bar{c}_{\sigma,i}\whole]
				+ [c_{\sigma}^{E'}(t) - \bar{c}_{\sigma,i}\whole] &= 0
		\end{align}
	\end{subequations}
	where $E'$ is the complementary band of $E \in \{A,B\}$.
	This agrees with \cite[Eq. (15)]{Morf:2006:may}.
	Finally, the difference of concentration between electrodes
	can be reduced to Eq. (\ref{idae:eqn:cE':intC}) since
	\begin{equation}
		c_{\sigma}^{E}(t) - c_{\sigma}^{E'}(t)
		= [c_{\sigma}^{E}(t) - \bar{c}_{\sigma,i}\whole]
		- [c_{\sigma}^{E'}(t) - \bar{c}_{\sigma,i}\whole]
	\end{equation}
	
	To obtain the limits for the difference of concentration in steady state,
	one considers Eq. \eqref{idae:eqn:average:intC} with non-negative concentrations on all electrodes
	\begin{subequations}
		\begin{align}
			-\bar{c}_{\sigma,i}\whole \leq [c_{\sigma,f}^{E} - \bar{c}_{\sigma,i}\whole]
			&= -[c_{\sigma,f}^{E'} - \bar{c}_{\sigma,i}\whole] \leq \bar{c}_{\sigma,i}\whole
			\\ {}
			-\bar{c}_{\sigma',i}\whole \leq [c_{\sigma',f}^{E} - \bar{c}_{\sigma',i}\whole]
			&= -[c_{\sigma',f}^{E'} - \bar{c}_{\sigma',i}\whole] \leq \bar{c}_{\sigma',i}\whole
		\end{align}
	\end{subequations}
	where $\sigma'$ is the complementary species of $\sigma \in \{O,R\}$.
	The limits for the concentration of species $\sigma'$
	may also affect the limits for the concentration of species $\sigma$.
	This can be seen by applying Corollary \ref{coplanar:cor:cf:total}
	(with $p_{x} = 2W$, $\Delta\bar{c}_{f}\unit = 0$
	and $\bar{c}_{\sigma,i}\unit = \bar{c}_{\sigma,i}\whole$)
	at the bands $E$ and $E'$
	\begin{subequations}
		\label{idae:eqn:cf:unit:total}
		\begin{align}
			D_{\sigma} [c_{\sigma,f}^{E} - \bar{c}_{\sigma,i}\whole]
			+ D_{\sigma'} [c_{\sigma',f}^{E} - \bar{c}_{\sigma',i}\whole] &= 0
			\\
			D_{\sigma} [c_{\sigma,f}^{E'} - \bar{c}_{\sigma,i}\whole]
			+ D_{\sigma'} [c_{\sigma',f}^{E'} - \bar{c}_{\sigma',i}\whole] &= 0
		\end{align}
	\end{subequations}
	and agrees with \cite[Eqs. (15) and (16)]{Morf:2006:may} when $D_{O} = D_{R}$.
	Combining the last two expressions leads to
	\begin{subequations}
		\begin{align}
			-D_{\sigma} \bar{c}_{\sigma,i}\whole
			\leq D_{\sigma} [c_{\sigma,f}^{E} - \bar{c}_{\sigma,i}\whole]
			&= -D_{\sigma} [c_{\sigma,f}^{E'} - \bar{c}_{\sigma,i}\whole]
			\leq D_{\sigma} \bar{c}_{\sigma,i}\whole
			\\ {}
			-D_{\sigma'} \bar{c}_{\sigma',i}\whole
			\leq D_{\sigma} [c_{\sigma,f}^{E} - \bar{c}_{\sigma,i}\whole]
			&= -D_{\sigma} [c_{\sigma,f}^{E'} - \bar{c}_{\sigma,i}\whole]
			\leq D_{\sigma'} \bar{c}_{\sigma',i}\whole
		\end{align}
	\end{subequations}
	Let $D_{\lambda} \bar{c}_{\lambda,i}\whole = \min(D_{O} \bar{c}_{O,i}\whole, D_{R} \bar{c}_{R,i}\whole)$,
	then the previous expressions can be summarized as
	\begin{equation}
		-D_{\lambda} \bar{c}_{\lambda,i}\whole
		\leq D_{\sigma} [c_{\sigma,f}^{E} - \bar{c}_{\sigma,i}\whole]
		= -D_{\sigma} [c_{\sigma,f}^{E'} - \bar{c}_{\sigma,i}\whole]
		\leq D_{\lambda} \bar{c}_{\lambda,i}\whole
	\end{equation}
	which leads to Eq. (\ref{idae:eqn:cotas:intC}).
\end{proof}

\begin{teorema}
	\label{idae:teo:cE-cE':extC}
	Consider an IDAE electrochemical cell under the assumptions of \S\ref{idae:promedios},
	but now with an external counter electrode as in Fig. \ref{idae:fig:cell:extC}.
	Assume also that the array of bands $E \in \{A,B\}$ is freely potentiostated,
	and its complementary array of bands $E'$ is also potentiostated,
	but fixed to a very extreme potential,
	such that the concentration of $\sigma \in \{O,R\}$ on its surface is $c_{\sigma}^{E'}(t) = 0$.
	
	If the IDAE has bands of equal width $2w_{A} = 2w_{B}$,
	the width of the counter electrode equals the width of the IDAE $w_{C} = 2W N_{E}$,
	and the integral $\int_{2W N_{E}}^{W_{T}-w_{C}} c_{\sigma}(x,0,t) - \bar{c}_{\sigma,i}\whole \ud{x} \approx 0$
	in the gap between the IDAE and the counter electrode,
	then the concentrations on the bands $c_{\sigma}^{E}(t)$
	and on the counter electrode $c_{\sigma}^{C}(t)$ are related for all $t$ by
	\begin{equation}
		\label{idae:eqn:cC:extC}
		\left[ \frac{c_{\sigma}^{E}(t)}{2} - \bar{c}_{\sigma,i}\whole \right]
		\approx -[c_{\sigma}^{C}(t) - \bar{c}_{\sigma,i}\whole]
	\end{equation}
	where $\bar{c}_{\sigma,i}\whole$ is the horizontal average
	in the initial steady state, defined in Eq. \eqref{idae:eqn:bar_ci:whole}.
	
	Moreover, when $\sigma = \lambda$, the difference of final concentrations
	is limited from above and below by
	\begin{subequations}
		\label{idae:eqn:cotas:extC}
		\begin{equation}
			0 
			\leq [c_{\lambda,f}^{E} - c_{\lambda,f}^{E'}] = c_{\lambda,f}^{E}
			\lesssim 4 \bar{c}_{\lambda,i}\whole
		\end{equation}
		and a similar situation occurs to its counter electrode,
		of which its concentration is limited from above and below by
		\begin{equation}
			0
			\leq c_{\lambda,f}^{C} 
			\lesssim 2 \bar{c}_{\lambda,i}\whole
		\end{equation}
	\end{subequations}
	where the determinant species $\lambda \in \{O,R\}$ is such that
	$D_{\lambda} \bar{c}_{\lambda,i}\whole = \min(D_{O} \bar{c}_{O,i}\whole, D_{R} \bar{c}_{R,i}\whole)$,
	$\lambda'$ is its complementary species,
	$c_{\lambda,f}^{E} = c_{\lambda}^{E}(+\infty)$,
	$c_{\lambda,f}^{E'} = c_{\lambda}^{E'}(+\infty)$
	and $c_{\lambda,f}^{C} = c_{\lambda}^{C}(+\infty)$.
	These last two expressions determine the limiting current of the cell
	when the electrochemical species satisfy
	\begin{equation}
		\label{idae:eqn:lambda':extC}
		D_{\lambda'} \bar{c}_{\lambda',i}\whole
		\geq 3 D_{\lambda} \bar{c}_{\lambda,i}\whole
	\end{equation}
\end{teorema}

Notice that Eqs. \eqref{idae:eqn:cotas:extC} are valid in steady state,
after a long time, comparable with the time constant of the whole cell.
See Eq. \eqref{idae:eqn:tau:intC} but replacing $W$ by $W_{T}$,
since the counter electrode is \emph{internal} to the whole cell.
Thus, for shallow cells $H \ll W_{T}$,
the dominant exponential mode is reduced to Eq. \eqref{idae:eqn:tau:extC}.
This agrees with the dominant mode of the unit cell with external counter electrode,
and only depends on the restricted height $H$.
Conversely, for very tall cells $H \gg W_{T}$,
Eq. \eqref{idae:eqn:tau:intC} depends only on the width of the whole cell $W_{T}$,
which produces much slower time responses.

Eqs. \eqref{idae:eqn:cotas:extC} can hold also $\forall t$
under the additional restriction $D_{O} = D_{R}$
(when applying Corollary \ref{coplanar:cor:ct:total}
instead of Corollary \ref{coplanar:cor:cf:total}
in Eqs. \eqref{idae:eqn:cf:whole:total}).

Finally, note that Eqs. (\ref{idae:eqn:cotas:extC}) imply the necessity of
the simultaneous presence of $\bar{c}_{O,i}\whole$ and $\bar{c}_{R,i}\whole$
in order to achieve steady state currents.

\begin{proof}
	Due to  Eqs. \eqref{idae:eqn:bar_c:whole},
	the average concentration at the bottom of the whole cell
	is given for all $t$ by
	\begin{subequations}
		\begin{align}
			\int_{0}^{W_{T}} c_{\sigma}(x,0,t) - \bar{c}_{\sigma,i}\whole \ud{x} &= 0
			\\
			\intertext{This integral can be splitted in three parts}
			\int_{0}^{2W N_{E}} +
			\underbrace{\int_{2W N_{E}}^{W_{T} - w_{C}}}_{\text{assumed}\approx 0} +
			\int_{W_{T} - w_{C}}^{W_{T}} &= 0
		\end{align}
		and when $2w_{A} = 2w_{B}$
		and the number of bands $N_{E}$ is large enough,
		it can be reduced to
		\begin{equation}
			\label{idae:eqn:pre-cC:extC}
			2W N_{E} \left[
				\frac{c_{\sigma}^{E}(t) + c_{\sigma}^{E'}(t)}{2} - \bar{c}_{\sigma,i}\whole
			\right]
			+ w_{C} [c_{\sigma}^{C}(t) - \bar{c}_{\sigma,i}\whole] \approx 0
		\end{equation}
	\end{subequations}
	since the average concentration on the IDAE satisfies
	\begin{equation}
		\frac{1}{2W N_{E}} \int_{0}^{2W N_{E}} c_{\sigma}(x,0,t) \ud{x}
		\approx \frac{c_{\sigma}^{E}(t) + c_{\sigma}^{E'}(t)}{2}
	\end{equation}
	Therefore, Eq. (\ref{idae:eqn:cC:extC}) is obtained from Eq. (\ref{idae:eqn:pre-cC:extC}),
	when $c_{\sigma}^{E'}(t) = 0$ and $2W N_{E} = w_{C}$.
	
	To obtain the limits for the difference of concentration in steady state,
	one considers Eq. \eqref{idae:eqn:cC:extC} with non-negative concentrations on all electrodes
	\begin{subequations}
		\begin{align}
		-\bar{c}_{\sigma,i}\whole
		\leq [c_{\sigma,f}^{E}/2 - \bar{c}_{\sigma,i}\whole]
		&\approx -[c_{\sigma,f}^{C} - \bar{c}_{\sigma,i}\whole]
		\leq \bar{c}_{\sigma,i}\whole
		\\
		-\bar{c}_{\sigma',i}\whole
		\leq [c_{\sigma',f}^{E}/2 - \bar{c}_{\sigma',i}\whole]
		&\approx -[c_{\sigma',f}^{C} - \bar{c}_{\sigma',i}\whole]
		\leq \bar{c}_{\sigma',i}\whole
		\end{align}
	\end{subequations}
	at the final steady state. Note that the limits for the concentration of species $\sigma'$ may also affect
	the limits for the concentration of species $\sigma$.
	This can be seen by applying the relations
	\begin{subequations}
		\label{idae:eqn:cf:whole:total}
		\begin{align}
			D_{\sigma} [c_{\sigma,f}^{E}/2 - \bar{c}_{\sigma,i}\whole]
			+ D_{\sigma'} [c_{\sigma',f}^{E}/2 - \bar{c}_{\sigma',i}\whole] &= 0
			\\
			D_{\sigma} [c_{\sigma,f}^{C} - \bar{c}_{\sigma,i}\whole]
			+ D_{\sigma'} [c_{\sigma',f}^{C} - \bar{c}_{\sigma',i}\whole] &= 0
		\end{align}
	\end{subequations}
	which are obtained from Corollary \ref{coplanar:cor:cf:total}
	(with $p_{x} = 2W_{T}$ and $\Delta\bar{c}_{f}\whole = 0$)
	at the middle of the gap between consecutive bands and on the surface of the counter electrode $C$.
	Then this leads to
	\begin{subequations}
		\begin{align}
		-D_{\sigma} \bar{c}_{\sigma,i}\whole
		\leq D_{\sigma} [c_{\sigma,f}^{E}/2 - \bar{c}_{\sigma,i}\whole]
		&\approx -D_{\sigma} [c_{\sigma,f}^{C} - \bar{c}_{\sigma,i}\whole]
		\leq D_{\sigma} \bar{c}_{\sigma,i}\whole
		\\
		-D_{\sigma'} \bar{c}_{\sigma',i}\whole
		\leq D_{\sigma} [c_{\sigma,f}^{E}/2 - \bar{c}_{\sigma,i}\whole]
		&\approx -D_{\sigma} [c_{\sigma,f}^{C} - \bar{c}_{\sigma,i}\whole]
		\leq D_{\sigma'} \bar{c}_{\sigma',i}\whole
		\end{align}
	\end{subequations}
	Let $D_{\lambda} \bar{c}_{\lambda,i}\whole = \min(D_{O} \bar{c}_{O,i}\whole, D_{R} \bar{c}_{R,i}\whole)$,
	then the previous expressions can be rewritten as
	\begin{equation}
		\label{idae:eqn:cotas:compacto}
		-D_{\lambda} \bar{c}_{\lambda,i}\whole
		\leq D_{\sigma} [c_{\sigma,f}^{E}/2 - \bar{c}_{\sigma,i}\whole]
		\approx -D_{\sigma} [c_{\sigma,f}^{C} - \bar{c}_{\sigma,i}\whole]
		\leq D_{\lambda} \bar{c}_{\lambda,i}\whole
	\end{equation}
	which finally leads to Eqs. (\ref{idae:eqn:cotas:extC}) when $\sigma = \lambda$.
	
	The reason to restrict Eqs. (\ref{idae:eqn:cotas:extC}) only to the determinant species $\lambda$
	is because only the species with least $D_{\sigma} \bar{c}_{\sigma,i}\whole$ 
	should be fixed to zero concentration $c_{\lambda,f}^{E'} = 0$ at the bands $E'$.
	Otherwise, non-negative concentrations are found in the equations, which are not possible physically.

	This can be seen by checking the minimum and maximum values for $D_{\sigma'} c_{\sigma',f}^{E}$.
	First, take the weighted sum of concentrations in Corollary \ref{coplanar:cor:cf:total}
	($p_{x} = 2W_{T}$ and $\Delta\bar{c}_{f}\whole = 0$) at the bands $E$
	\begin{subequations}
		\begin{equation}
			D_{\sigma} c_{\sigma,f}^{E} + D_{\sigma'} c_{\sigma',f}^{E}
			= D_{\lambda} \bar{c}_{\lambda,i}\whole
			+ D_{\lambda'} \bar{c}_{\lambda',i}\whole
		\end{equation}
		where $\lambda'$ is the complementary species of $\lambda$.
		At the same time that $D_{\sigma} c_{\sigma,f}^{E}$ reaches its extrema, given by Eq. \eqref{idae:eqn:cotas:compacto}
		\begin{equation}
			2 D_{\sigma} \bar{c}_{\sigma,i}\whole - 2 D_{\lambda} \bar{c}_{\lambda,i}\whole
			\leq D_{\sigma} c_{\sigma,f}^{E} \lesssim
			2 D_{\sigma} \bar{c}_{\sigma,i}\whole + 2 D_{\lambda} \bar{c}_{\lambda,i}\whole
		\end{equation}
		$D_{\sigma'} c_{\sigma',f}^{E}$ also reaches its extrema, which are given by
		\begin{equation}
			D_{\lambda'} \bar{c}_{\lambda',i}\whole
			- D_{\lambda} \bar{c}_{\lambda,i}\whole
			-2 D_{\sigma} \bar{c}_{\sigma,i}\whole
			\lesssim D_{\sigma'} c_{\sigma',f}^{E}
			\leq
			3 D_{\lambda} \bar{c}_{\lambda,i}\whole
			+ D_{\lambda'} \bar{c}_{\lambda',i}\whole
			- 2 D_{\sigma} \bar{c}_{\sigma,i}\whole
		\end{equation}
	\end{subequations}
	Since the minimun concentration should be non-negative,
	this leads to the fact that $\sigma$ must be only $\lambda$ (and not $\lambda'$),
	and also to the additional restriction in Eq. (\ref{idae:eqn:lambda':extC}).
\end{proof}

Theorems \ref{idae:teo:cE-cE':intC} and \ref{idae:teo:cE-cE':extC}
show the usefulness of the properties of horizontal averages in Corollaries
\ref{coplanar:cor:bar_ci}, \ref{coplanar:cor:delta_bar_c} and \ref{coplanar:cor:bar_cf},
since they provide a tool for determining the concentration on the counter electrode,
which is unknown \emph{a priori}, since it is controlled by the potentiostat.
This is useful for determining boundary conditions for the counter electrode
and it can be used in simulations that require its inclussion.

Interesting is the fact that having an internal (Theorem \ref{idae:teo:cE-cE':intC})
or external (Theorem \ref{idae:teo:cE-cE':extC}) counter electrode
produces bipolar or unipolar limiting currents respectively.

Also, under the conditions that were just analyzed: 
band electrodes of equal width $2w_{A} = 2w_{B}$ 
and external counter electrode of width equal to that of the IDAE $w_{C} = 2W N_{E}$,
it appears to be that there is no significant advantage
of having an external counter electrode (both electrode arrays potentiostated)
versus using one of the arrays as counter electrode (only one array potentiostated),
at least in terms of the current range.
According to Lemma \ref{idae:lem:if} and Theorem \ref{idae:teo:cE-cE':extC} for external counter electrode,
the current is $i_{f}^{E} \propto c_{\lambda,f}^{E} - c_{\lambda,f}^{E'}$
and the difference of concentration ranges from $0$ to $+4\bar{c}_{\lambda,i}\whole$.
On the other hand, according to Lemma \ref{idae:lem:if} and Theorem \ref{idae:teo:cE-cE':intC} for internal counter electrode,
the current is still $i_{f}^{E} \propto c_{\lambda,f}^{E} - c_{\lambda,f}^{E'}$
but the difference of concentration ranges from $-2\bar{c}_{\lambda,i}\whole$ to $+2\bar{c}_{\lambda,i}\whole$, that is, it also spans a range of $4\bar{c}_{\lambda,i}\whole$.

Nevertheless, taking a careful look to the average at the bottom of the whole cell in Eq. \eqref{idae:eqn:pre-cC:extC} 
suggests that with an external counter electrode wider than the IDAE $w_{C} \geq 2W N_{E}$,
the current could span larger ranges.
This is because a slight decrease of the concentration on the counter electrode $c_{\lambda,f}^{C}$
below the average concentration $\bar{c}_{\lambda,i}\whole$
could cause a large increase of the concentration of the freely potentiostated array $c_{\lambda,f}^{E}$.
However, the analysis of the non-linearities
(related to the physical constraint of non-negative concentrations)
and the limits for the steady state current become more difficult.


\section{Conclusions}

The properties of horizontal average and (weighted) sum of concentrations
show several implications in the behavior of
a periodic cell with finite height and two-dimensional symmetry.

In the initial steady state,
the net current is zero (100\% collection efficiency) and both,
the horizontal average and the (weighted) sum of concentrations of both species,
are uniform in the cell, see Corollaries \ref{coplanar:cor:bar_ci} and \ref{coplanar:cor:ci:total}.

During the transtient state,
a change in the average of concentration at the bottom of the cell
(where the electrodes are located) produces a non-zero net current
(collection efficiency less than 100\%),
see Eq. \eqref{coplanar:eqn:i_net} in Corollary \ref{coplanar:cor:delta_bar_c}.
Subsequently, this non-zero net current during the transient
produces accumulation (or depletion) of species in the cell,
see Eq. \eqref{coplanar:eqn:delta_bar_c}
in Corollary \ref{coplanar:cor:delta_bar_c}.
The duration of the transient is governed by the slowest exponential mode
in Eq. \eqref{coplanar:eqn:i_net}, and increases as the height $H$ increases.

In the final steady state, the net current must be zero 
(100\% collection efficiency)
despite any non-zero net current during the transient.
The horizontal average and (weighted) sum of concentrations
become again uniform in the cell, but may not equal their initial counterparts
due to accumulation (or depletion) of species.
See Corollaries \ref{coplanar:cor:bar_cf} and \ref{coplanar:cor:cf:total}.

Note that this accumulation (or depletion) of species
is not present during the transient and final steady states,
when the average of concentration at the bottom of the cell
remains at the same value as its initial counterpart.
In this case the net current always equals zero,
therefore the horizontal average and (weighted) sum of concentrations
maintain the same value in the initial and final steady states.

If the cell has semi-infinite geometry $H \to +\infty$,
the final steady state behaves differently.
When the average concentration at the bottom of the cell
is driven out from its initial counterpart,
the final net current becomes non-zero (collection efficiency less than 100\%) and the horizontal average of final concentration loses its uniformity along the $z$-axis.

These properties of horizontal averages and (weighted) sum of concentrations
are also useful for determining the concentration on a counter electrode
(which is not known a priori), and to determine non-linearities
caused by depletion of electrochemical species
at electrodes potentiostated at extreme voltages.
This can be seen for the case of IDAE
in Theorems \ref{idae:teo:cE-cE':intC} and \ref{idae:teo:cE-cE':extC},
and it is specially important to take into account when performing simulations.

More results have been found for IDAE using the previous properties.
Normally, the IDAE is operated in dual mode (voltages are applied at each array)
either with an external or internal counter electrode,
of which the former is most commonly found in the literature.
Comparing both modes, the results
in Lema \ref{idae:lem:if} and Theorems \ref{idae:teo:cE-cE':intC} and \ref{idae:teo:cE-cE':extC}
show that the maximum current range that can be spanned in steady state,
either when using external or internal counter electrode, is the same.
Also, it is shown that the time the current requires to reach steady state
in case of using external counter electrode is longer
than that required in case of using internal counter electrode,
see Eqs. \eqref{idae:eqn:tau:extC} and \eqref{idae:eqn:tau:intC} respectively.
This suggests that, despite of being more common in the literature,
an IDAE configuration with external counter electrode provides
no significant advantage compared with the case of internal counter electrode.
This is true, however, under the restrictions
of Theorems \ref{idae:teo:cE-cE':intC} and \ref{idae:teo:cE-cE':extC},
where the IDAE has bands of equal width
and the width of the counter electrode equals that of the whole IDAE.


\section*{Acknowledgements}

The authors would like to thank
\emph{King Mongkut’s University of Technology Thonburi} through the
\emph{Petchra Pra Jom Klao Ph.~D. scholarship} (Grant No. 28/2558)
for financially supporting {CFGY}.
Finally, the authors acknowledge 
the \emph{Higher Education Research	Promotion} and \emph{National Research University Project of Thailand}, \emph{Office of the Higher Education Commission} 
and the financial support provided by \emph{King Mongkut's University of Technology Thonburi} through the \emph{KMUTT 55th Anniversary Commemorative Fund}.

\bibliographystyle{unsrtnat}
\bibliography{references}

\begin{thebibliography}{27}
\providecommand{\natexlab}[1]{#1}
\providecommand{\url}[1]{\texttt{#1}}
\expandafter\ifx\csname urlstyle\endcsname\relax
  \providecommand{\doi}[1]{doi: #1}\else
  \providecommand{\doi}{doi: \begingroup \urlstyle{rm}\Url}\fi

\bibitem[Dayton et~al.(1980)Dayton, Brown, Stutts, and Wightman]{Dayton:1980:}
M.~A. Dayton, J.~C. Brown, K.~J. Stutts, and R.~M. Wightman.
\newblock Faradaic electrochemistry at microvoltammetric electrodes.
\newblock \emph{Analytical Chemistry}, 52\penalty0 (6):\penalty0 946--950,
  1980.
\newblock \doi{10.1021/ac50056a040}.

\bibitem[Forster and Keyes(2007)]{Forster:2007:}
Robert~J. Forster and Tia~E. Keyes.
\newblock \emph{Behavior of ultramicroelectrodes}, chapter 6.1, pages 155 --
  171.
\newblock In  \citet{Zoski:2007:}, 1 edition, 2007.
\newblock ISBN 978-0-444-51958-0.
\newblock \doi{10.1016/B978-044451958-0.50007-0}.

\bibitem[Szunerits and Thouin(2007)]{Szunerits:2007:}
Sabine Szunerits and Laurent Thouin.
\newblock \emph{Microelectrode Arrays}, chapter~10, pages 391 -- XI.
\newblock In  \citet{Zoski:2007:}, 1 edition, 2007.
\newblock ISBN 978-0-444-51958-0.
\newblock \doi{10.1016/B978-044451958-0.50023-9}.

\bibitem[Morf(1996)]{Morf:1996:sep}
Werner~E. Morf.
\newblock Theoretical treatment of the amperometric current response of
  multiple microelectrode arrays.
\newblock \emph{Analytica Chimica Acta}, 330\penalty0 (2):\penalty0 139 -- 149,
  September 1996.
\newblock ISSN 0003-2670.
\newblock \doi{https://doi.org/10.1016/0003-2670(96)00148-1}.

\bibitem[Morf et~al.(2006)Morf, Koudelka-Hep, and de~Rooij]{Morf:2006:may}
Werner~E. Morf, Milena Koudelka-Hep, and Nicolaas~F. de~Rooij.
\newblock Theoretical treatment and computer simulation of microelectrode
  arrays.
\newblock \emph{J. Electroanal. Chem.}, 590\penalty0 (1):\penalty0 47--56, May
  2006.
\newblock ISSN 15726657.
\newblock \doi{10.1016/j.jelechem.2006.01.028}.

\bibitem[Aoki et~al.(1988)Aoki, Morita, Niwa, and Tabei]{Aoki:1988:dec}
Koichi Aoki, Masao Morita, Osamu Niwa, and Hisao Tabei.
\newblock Quantitative analysis of reversible diffusion-controlled currents of
  redox soluble species at interdigitatedgitated array electrodes under
  steady-state conditions.
\newblock \emph{J. Electroanal. Chem. Interfacial Electrochem.}, 256\penalty0
  (2):\penalty0 269--282, December 1988.
\newblock ISSN 00220728.
\newblock \doi{10.1016/0022-0728(88)87003-7}.

\bibitem[Aoki(1990)]{Aoki:1990:apr}
Koichi Aoki.
\newblock Theory of stationary current-potential curves at interdigitated
  microarray electrodes for quasi-reversible and totally irreversible electrode
  reactions.
\newblock \emph{Electroanalysis}, 2\penalty0 (3):\penalty0 229--233, April
  1990.
\newblock ISSN 1040-0397.
\newblock \doi{10.1002/elan.1140020310}.

\bibitem[Bard et~al.(1986)Bard, Crayston, Kittlesen, Shea, and
  Wrighton]{Bard:1986:sep}
Allen~J. Bard, Joseph~A. Crayston, Gregg~P. Kittlesen, Theresa~Varco Shea, and
  Mark~S. Wrighton.
\newblock Digital simulation of the measured electrochemical response of
  reversible redox couples at microelectrode arrays: consequences arising from
  closely spaced ultramicroelectrodes.
\newblock \emph{Analytical Chemistry}, 58\penalty0 (11):\penalty0 2321--2331,
  September 1986.
\newblock \doi{10.1021/ac00124a045}.

\bibitem[Streeter et~al.(2007)Streeter, Fietkau, del Campo, Mas, Muñoz, and
  Compton]{Streeter:2007:aug}
Ian Streeter, Nicole Fietkau, Javier del Campo, Roser Mas, Francesc~Xavier
  Muñoz, and Richard~G. Compton.
\newblock Voltammetry at regular microband electrode arrays: Theory and
  experiment.
\newblock \emph{The Journal of Physical Chemistry C}, 111\penalty0
  (32):\penalty0 12058--12066, August 2007.
\newblock \doi{10.1021/jp073224d}.

\bibitem[Pebay et~al.(2013)Pebay, Sella, Thouin, and Amatore]{Pebay:2013:dec}
Cécile Pebay, Catherine Sella, Laurent Thouin, and Christian Amatore.
\newblock Mass transport at infinite regular arrays of microband electrodes
  submitted to natural convection: Theory and experiments.
\newblock \emph{Analytical Chemistry}, 85\penalty0 (24):\penalty0 12062--12069,
  December 2013.
\newblock \doi{10.1021/ac403159j}.

\bibitem[Aoki and Tanaka(1989)]{Aoki:1989:jul}
Koichi Aoki and Mitsuya Tanaka.
\newblock Time-dependence of diffusion-controlled currents of a soluble redox
  couple at interdigitated microarray electrodes.
\newblock \emph{Journal of Electroanalytical Chemistry}, 266\penalty0
  (1):\penalty0 11--20, July 1989.
\newblock ISSN 00220728.
\newblock \doi{10.1016/0022-0728(89)80211-6}.

\bibitem[Jin et~al.(1996)Jin, Qian, Zhang, and Shi]{Jin:1996:aug:b}
Baokang Jin, Weijun Qian, Zuxun Zhang, and Hansheng Shi.
\newblock Application of the finite analytic numerical method. part 1.
  diffusion problems on coplanar and elevated interdigitated microarray band
  electrodes.
\newblock \emph{Journal of Electroanalytical Chemistry}, 411\penalty0
  (1-2):\penalty0 29--36, August 1996.
\newblock \doi{10.1016/0022-0728(96)04594-9}.

\bibitem[Yang and Zhang(2007)]{Yang:2007:oct}
Xiaoling Yang and Guigen Zhang.
\newblock The voltammetric performance of interdigitated electrodes with
  different electron-transfer rate constants.
\newblock \emph{Sensors and Actuators B: Chemical}, 126\penalty0 (2):\penalty0
  624--631, October 2007.
\newblock \doi{10.1016/j.snb.2007.04.013}.

\bibitem[Bieniasz(2015)]{Bieniasz:2015:oct}
L.K. Bieniasz.
\newblock Theory of potential step chronoamperometry at a microband electrode:
  Complete explicit semi-analytical formulae for the faradaic current density
  and the faradaic current.
\newblock \emph{Electrochimica Acta}, 178\penalty0 (Supplement C):\penalty0 25
  -- 33, October 2015.
\newblock ISSN 0013-4686.
\newblock \doi{10.1016/j.electacta.2015.07.040}.

\bibitem[Han et~al.(2014)Han, Kim, Kang, and Chung]{Han:2014:}
Donghoon Han, Yang-Rae Kim, Chung~Mu Kang, and Taek~Dong Chung.
\newblock Electrochemical signal amplification for immunosensor based on 3d
  interdigitated array electrodes.
\newblock \emph{Analytical Chemistry}, 86\penalty0 (12):\penalty0 5991--5998,
  2014.
\newblock \doi{10.1021/ac501120y}.
\newblock PMID: 24842332.

\bibitem[Kanno et~al.(2014)Kanno, Goto, Ino, Inoue, Takahashi, Shiku, and
  Matsue]{Kanno:2014:}
Yusuke Kanno, Takehito Goto, Kosuke Ino, Kumi~Y Inoue, Yasufumi Takahashi,
  Hitoshi Shiku, and Tomokazu Matsue.
\newblock Su-8-based flexible amperometric device with ida electrodes to
  regenerate redox species in small spaces.
\newblock \emph{Analytical Sciences}, 30\penalty0 (2):\penalty0 305--309, 2014.
\newblock \doi{10.2116/analsci.30.305}.

\bibitem[Strutwolf and Williams(2005)]{Strutwolf:2005:feb}
Jörg Strutwolf and D.~E. Williams.
\newblock Electrochemical sensor design using coplanar and elevated
  interdigitated array electrodes. a computational study.
\newblock \emph{Electroanalysis}, 17\penalty0 (2):\penalty0 169--177, February
  2005.
\newblock ISSN 1040-0397.
\newblock \doi{10.1002/elan.200403112}.

\bibitem[Goluch et~al.(2009)Goluch, Wolfrum, Singh, Zevenbergen, and
  Lemay]{Goluch:2009:may}
Edgar~D Goluch, Bernhard Wolfrum, Pradyumna~S Singh, Marcel A~G Zevenbergen,
  and Serge~G Lemay.
\newblock Redox cycling in nanofluidic channels using interdigitated
  electrodes.
\newblock \emph{Analytical and Bioanalytical Chemistry}, 394\penalty0
  (2):\penalty0 447--56, May 2009.
\newblock ISSN 1618-2650.
\newblock \doi{10.1007/s00216-008-2575-x}.

\bibitem[Bellagha-Chenchah et~al.(2016)Bellagha-Chenchah, Sella, and
  Thouin]{Bellagha-Chenchah:2016:jun}
Wided Bellagha-Chenchah, Catherine Sella, and Laurent Thouin.
\newblock Understanding mass transport at channel microband electrodes:
  Influence of confined space under stagnant conditions.
\newblock \emph{Electrochimica Acta}, 202:\penalty0 122--130, June 2016.
\newblock \doi{10.1016/j.electacta.2016.04.011}.

\bibitem[Guajardo et~al.(2013)Guajardo, Ngamchana, and
  Surareungchai]{GuajardoYevenes:2013:sep}
Cristian Guajardo, Sirimarn Ngamchana, and Werasak Surareungchai.
\newblock Mathematical modeling of interdigitated electrode arrays in finite
  electrochemical cells.
\newblock \emph{Journal of Electroanalytical Chemistry}, 705:\penalty0 19--29,
  September 2013.
\newblock ISSN 15726657.
\newblock \doi{10.1016/j.jelechem.2013.07.014}.

\bibitem[Shim et~al.(2013)Shim, Rust, and Ahn]{Shim:2013:}
Joon~S. Shim, Michael~J. Rust, and Chong~H. Ahn.
\newblock A large area nano-gap interdigitated electrode array on a polymer
  substrate as a disposable nano-biosensor.
\newblock \emph{Journal of Micromechanics and Microengineering}, 23\penalty0
  (3):\penalty0 035002, 2013.
\newblock \doi{10.1088/0960-1317/23/3/035002}.

\bibitem[Olver et~al.()Olver, {Olde Daalhuis}, Lozier, Schneider, Boisvert,
  Clark, Miller, and Saunders]{dlmf}
F.~W.~J. Olver, A.~B. {Olde Daalhuis}, D.~W. Lozier, B.~I. Schneider, R.~F.
  Boisvert, C.~W. Clark, B.~R. Miller, and B.~V. Saunders, editors.
\newblock \emph{NIST Digital Library of Mathematical Functions}.
\newblock URL \url{http://dlmf.nist.gov/}.
\newblock Release 1.0.17 of 2017-12-22.

\bibitem[Guyer et~al.(2009)Guyer, Wheeler, and Warren]{Guyer:2009:may}
Jonathan~E. Guyer, Daniel Wheeler, and James~A. Warren.
\newblock Fipy: Partial differential equations with python.
\newblock \emph{Computing in Science \& Engineering}, 11\penalty0 (3):\penalty0
  6--15, May 2009.
\newblock ISSN 1521-9615.
\newblock \doi{10.1109/mcse.2009.52}.

\bibitem[Britz and Strutwolf(2016)]{Britz:2016:}
Dieter Britz and Jörg Strutwolf.
\newblock \emph{Digital Simulation in Electrochemistry}.
\newblock Monographs in Electrochemistry. Springer International Publishing, 4
  edition, 2016.
\newblock ISBN 978-3-319-30290-4, 978-3-319-30292-8.
\newblock \doi{10.1007/978-3-319-30292-8}.

\bibitem[Ablowitz and Fokas(2003)]{Ablowitz:2003:apr}
Mark~J. Ablowitz and Athanassios~S. Fokas.
\newblock \emph{Complex variables: Introduction and applications}.
\newblock Cambridge University Press, 2nd edition, April 2003.

\bibitem[Olver(2017)]{Olver:2017:}
Peter~J Olver.
\newblock Complex analysis and conformal mapping, 2017.
\newblock URL \url{http://www.math.umn.edu/~olver/ln_/cml.pdf}.

\bibitem[Zoski(2007)]{Zoski:2007:}
Cynthia~G. Zoski.
\newblock \emph{Handbook of Electrochemistry}.
\newblock Elsevier, Amsterdam, 1 edition, 2007.
\newblock ISBN 978-0-444-51958-0.
\newblock \doi{10.1016/b978-0-444-51958-0.x5000-9}.

\end{thebibliography}


\clearpage
\begin{center}
	\LARGE
	Supplementary information
\end{center}

\setcounter{page}{1}
\setcounter{section}{0}
\renewcommand*{\thesection}{S\arabic{section}}  
\renewcommand*{\theHsection}{\thesection}  


\section{Additional definitions}

\subsection{Fourier series and Laplace transform}
\label{transforms}

Fourier series and Laplace transform are extensively used in \S\ref{theory},
therefore they are defined briefly in this section.

\begin{definicion}
	The exponential version of the Fourier coefficients $\fourier f(n_{x})$
	and Fourier series of the periodic function $f(x)$, with period $p_{x}$,
	are defined by
	\begin{subequations}
		\label{transforms:eqn:fourier}
		\begin{align}
			\fourier f(n_{x}) &= \frac{1}{p_{x}}
			\int_{-p_{x}/2}^{+p_{x}/2} f(x)\, \e^{-\bm{i} x\, n_{x} 2\pi/p_{x}} \ud{x}
			\\
			f(x) &= \sum_{n_{x} = -\infty}^{+\infty} \fourier f(n_{x})\,
			\e^{\bm{i} x\, n_{x} 2\pi/p_{x}}
		\end{align}
	\end{subequations}
	where $\bm{i}$ is the imaginary unit.
\end{definicion}

\begin{definicion}
	The Laplace transform $\laplace f(s)$ of the time function $f(t)$,
	and its inverse, are defined respectively as
	\begin{subequations}
		\begin{align}
			\laplace f(s) &= \int_{0^{-}}^{+\infty} f(t)\, \e^{-st} \ud{t}
			\\
			f(t) &= \frac{1}{\bm{i} 2\pi} \oint_{\Gamma} \laplace f(s)\, \e^{st} \ud{s}
		\end{align}
	\end{subequations}
	where $\bm{i}$ is the imaginary unit,
	and $\Gamma$ is a closed path in the complex plane,
	surrounding the poles of $\laplace f(s)$.
\end{definicion}

\section{Output data of simulations}

\subsection{Figure \ref{idae:fig:mesh}}
\label{data:s32:mesh}
\begin{lstlisting}[basicstyle={\ttfamily \footnotesize}]
H: 1.0
------------------------------
d0: 0.0005
r: 1.2482936464
R: 398.613975357
(nx, ny): (88, 28)
iA_sim: -0.49598704973
iB_sim: 0.495987042278
error A: -0.49598704973
error B: 0.495987042278
------------------------------
d0: 0.00025
r: 1.24958205268
R: 799.729964567
(nx, ny): (100, 31)
iA_sim: -0.496126563111
iB_sim: 0.496126585387
error A: -0.000139513380364
error B: 0.000139543108393
------------------------------
d0: 0.000125
r: 1.2506259256
R: 1604.00273474
(nx, ny): (112, 34)
iA_sim: -0.496191977555
iB_sim: 0.496192018328
error A: -6.54144439734e-05
error B: 6.54329414557e-05
------------------------------
d0: 6.25e-05
r: 1.25148418917
R: 3215.97912429
(nx, ny): (124, 37)
iA_sim: -0.496220197796
iB_sim: 0.49622043774
error A: -2.82202415637e-05
error B: 2.84194122416e-05
------------------------------

H: 0.5
------------------------------
d0: 0.00025
r: 1.2482936464
R: 398.613975357
(nx, ny): (100, 28)
iA_sim: -0.460181970757
iB_sim: 0.460181981376
error A: -0.460181970757
error B: 0.460181981376
------------------------------
d0: 0.000125
r: 1.24958205268
R: 799.729964567
(nx, ny): (112, 31)
iA_sim: -0.460242109636
iB_sim: 0.460242169828
error A: -6.013887907e-05
error B: 6.0188452456e-05
------------------------------
d0: 6.25e-05
r: 1.2506259256
R: 1604.00273474
(nx, ny): (124, 34)
iA_sim: -0.460268139168
iB_sim: 0.46026817421
error A: -2.60295319954e-05
error B: 2.60043817906e-05
------------------------------

H: 0.3
------------------------------
d0: 0.00015
r: 1.2482936464
R: 398.613975357
(nx, ny): (108, 28)
iA_sim: -0.377047535064
iB_sim: 0.377047625461
error A: -0.377047535064
error B: 0.377047625461
------------------------------
d0: 7.5e-05
r: 1.24958205268
R: 799.729964567
(nx, ny): (120, 31)
iA_sim: -0.377077543754
iB_sim: 0.377077632037
error A: -3.00086899314e-05
error B: 3.00065762477e-05
------------------------------

\end{lstlisting}

\subsection{Figure \ref{idae:fig:i_sim}}
\label{data:s32:simulation}
\begin{lstlisting}[basicstyle={\ttfamily \footnotesize}]
H: 1.0, ci: 0.25
------------------------------
d0: 6.25e-05
r: 1.25148418917
R: 3215.97912429
(nx, ny): (124, 37)

dt: 0.25
nt: 120
tf: 30.0

Dsim: 0.101321183642
iA_sim(tf): -0.495578547168
iB_sim(tf): 0.496862369362
+/- Dcf/[cfB - cfA] = 0.249351679207
------------------------------

H: 0.5, ci: 0.25
------------------------------
d0: 6.25e-05
r: 1.2506259256
R: 1604.00273474
(nx, ny): (124, 34)

dt: 0.25
nt: 120
tf: 30.0

Dsim: 0.101321183642
iA_sim(tf): -0.460265207649
iB_sim(tf): 0.460271093209
+/- Dcf/[cfB - cfA] = 0.250010806019
------------------------------

H: 0.3, ci: 0.25
------------------------------
d0: 7.5e-05
r: 1.24958205268
R: 799.729964567
(nx, ny): (120, 31)

dt: 0.25
nt: 120
tf: 30.0

Dsim: 0.101321183642
iA_sim(tf): -0.377075628497
iB_sim(tf): 0.377079666382
+/- Dcf/[cfB - cfA] = 0.250025381405
------------------------------

H: 1.0, ci: 0.5
------------------------------
d0: 6.25e-05
r: 1.25148418917
R: 3215.97912429
(nx, ny): (124, 37)

dt: 0.05
nt: 100
tf: 5.0

Dsim: 0.101321183642
iA_sim(tf): -0.496376671225
iB_sim(tf): 0.49637666095
+/- Dcf/[cfB - cfA] = 2.01325645856e-10
------------------------------

H: 0.5, ci: 0.5
------------------------------
d0: 6.25e-05
r: 1.2506259256
R: 1604.00273474
(nx, ny): (124, 34)

dt: 0.05
nt: 100
tf: 5.0

Dsim: 0.101321183642
iA_sim(tf): -0.460300294391
iB_sim(tf): 0.46030029814
+/- Dcf/[cfB - cfA] = -1.56061988381e-08
------------------------------

H: 0.3, ci: 0.5
------------------------------
d0: 7.5e-05
r: 1.24958205268
R: 799.729964567
(nx, ny): (120, 31)

dt: 0.05
nt: 100
tf: 5.0

Dsim: 0.101321183642
iA_sim(tf): -0.37707912606
iB_sim(tf): 0.377079134744
+/- Dcf/[cfB - cfA] = -1.38156042166e-08
------------------------------

\end{lstlisting}

\end{document}